%% file: main.tex
\documentclass[sigconf]{acmart}
\AtBeginDocument{%
  }

\setcopyright{acmlicensed}
\copyrightyear{2018}
\acmYear{2018}
\acmDOI{XXXXXXX.XXXXXXX}
\acmConference[Conference acronym 'XX]{Make sure to enter the correct
  conference title from your rights confirmation email}{June 03--05,
  2018}{Woodstock, NY}
\acmISBN{978-1-4503-XXXX-X/2018/06}

\usepackage[normalem]{ulem}
\usepackage{tikz}
\usepackage{enumitem}

\usepackage{stmaryrd}

\usepackage{tcolorbox}
\tcbuselibrary{skins,breakable}

\newtcolorbox{mybox}[2][]{%
  enhanced,
  title        = {#2},
  attach boxed title to top left={xshift=+3mm,yshift*=-3mm},
  colback      = white,
  colframe     = black,
  fonttitle    = \bfseries,
  fontupper    = \small,
  fontlower    = \small,
  colbacktitle = black!3!white,
  coltitle     = black,
  #1
}

\newcommand{\dom}[1]{\mathbb{#1}}

\def\out[#1]{\dom{Y}_{#1}}
\def\Aldp[#1]{\mathcal{R}_{#1}}

\newtheorem{theorem}{Theorem}[section]

\newtheorem{lemma}[theorem]{Lemma}

\newtheorem{definition}[theorem]{Definition}

\usepackage{pifont}

\usepackage[vlined, linesnumbered, ruled, nofillcomment]{algorithm2e}
\DontPrintSemicolon

\SetKwComment{tcp}{// }{}

\SetKwInOut{KwIn}{Input}
\SetKwInOut{KwOut}{Output}

\SetKwFunction{CF}{CF}
\SetKwFunction{CFInv}{CF^{-1}}

\usepackage{stmaryrd}
\usepackage{multirow}

\usepackage[dvipsnames, svgnames, x11names]{xcolor}

\newcommand{\rescale}{\text{\color{DarkGreen}Rescale}}
\newcommand{\expansion}{\text{\color{Tomato}Expansion}}

\newcommand{\Merge}{\text{\color{Blue}Merge}}

\newcommand{\revise}[1]{\begingroup #1 \endgroup}

\begin{document}

\title{Multi-tier Differential Private Query Release}


\author{Shaowei Wang, Jinn Li,\\ Yun Peng}
\affiliation{%
  \institution{Guangzhou University}
  \city{Guangzhou}
  \country{China}}
\email{wangsw@gzhu.edu.cn}

\author{Puning Zhao, Wenqi Ren}
\affiliation{%
  \institution{Sun Yat-sen University}
  \city{Shenzhen}
  \country{China}
}

\author{Changyu Dong, Jin Li,\\ Jian Weng}
\affiliation{%
 \institution{Guangzhou University}
 \city{Guangzhou}
 \country{China}}

\renewcommand{\shortauthors}{Wang et al.}

\begin{abstract}
\revise{Answering statistical queries over sensitive data under differential privacy (DP) is a common task in many settings, including databases, mobile computing, and data markets. In these scenarios, multiple analysts may issue the same query, while receiving answers generated under different privacy budgets due to differences in trust levels or willingness to pay. Existing approaches for such multi-tier DP queries either incur excessive cumulative privacy loss or suffer from suboptimal utility. In this paper, we propose a framework for multi-tier DP query release that simultaneously bound the cumulative privacy loss by the maximum privacy budget among all queries and achieve optimal utility comparable to that of single-tier mechanisms. Our framework applies to different classes of DP mechanisms. For noise-adding mechanisms (e.g., count queries with the two-sided Geometric mechanism in the curator model), we develop a general solution based on the characteristic functions of noise distributions. For other mechanisms (e.g., count queries under the local DP model with the Subset mechanism), we design mechanism-specific primitives for budget transformation and introduce a template-based strategy that attains optimal utility across different privacy regimes. Experimental results demonstrate the effectiveness of our framework.}
\end{abstract}





\maketitle

\input{body}


\bibliographystyle{ACM-Reference-Format}
\bibliography{refs}

\appendix
\input{appendix}

\end{document}

%% file: body.tex
\section{Introduction}\label{sec:intro}
Collecting and analyzing user data benefits a wide range of domains, from healthcare services to mobile computing. However, because user data in these domains often contains sensitive personal information, preserving data privacy is essential. Differential privacy (DP) \cite{dwork2006differential} has emerged as a de facto standard for privacy protection and is widely adopted in real-world industrial systems. For example, Google \cite{zhang2024differentially} and LinkedIn \cite{rogers2021linkedin} use DP for analytical queries in data warehouses, while Apple \cite{appleDP}, Google \cite{GoogleDP}, and Microsoft \cite{ding2017collecting} apply DP on the user side to collect mobile service usage data.

In many DP systems, multiple parties may issue the same query, often with varying privacy budgets and utility targets. For example:

\begin{itemize}[leftmargin=1em]
\item \emph{Multi-analyst data warehouses.} Different analysts within a data warehouse may query the same statistics, such as a user population histogram, with varying levels of trust and DP budgets. For example, LinkedIn's Audience Engagement API platform \cite{rogers2021linkedin} reports that internal analysts (such as for dashboards, anomaly detection, and A/B testing) and external analysts (e.g., those asking profile statistics) are allocated different privacy budgets.

    
\item \emph{Multi-target mobile services.} Mobile devices often run multiple Apps and services concurrently, many of which rely on user demographic information and usage data, albeit with varying levels of accuracy (e.g., with different local DP noises). For example, Google uses DP-protected keyboard typing data for both next-word prediction modeling \cite{GoogleDP2} and out-of-vocabulary (OOV) word discovery \cite{GoogleDP}; personalized services (e.g., Apple Maps) typically require relatively accurate user information to deliver high-quality experiences, whereas aggregate analytics services (e.g., Apple and Google's epidemic exposure notification~\cite{applegoogle}) may operate effectively with only approximate information.

\item \emph{Multi-buyer data markets.} In data markets \cite{azcoitia2022survey,alabi2025privacy}, data owners provide access to sensitive datasets to multiple buyers with different monetary budgets and accuracy requirements. Different buyers may purchase the same data asset (e.g., individual records or precomputed statistics) at varying levels of precision (e.g., controlled by DP noise) depending on their willingness to pay. This practice is known as data versioning \cite{yu2017data} in data markets.
\end{itemize}

\begin{figure*}[t]
\centering
\includegraphics[width=0.88\linewidth]{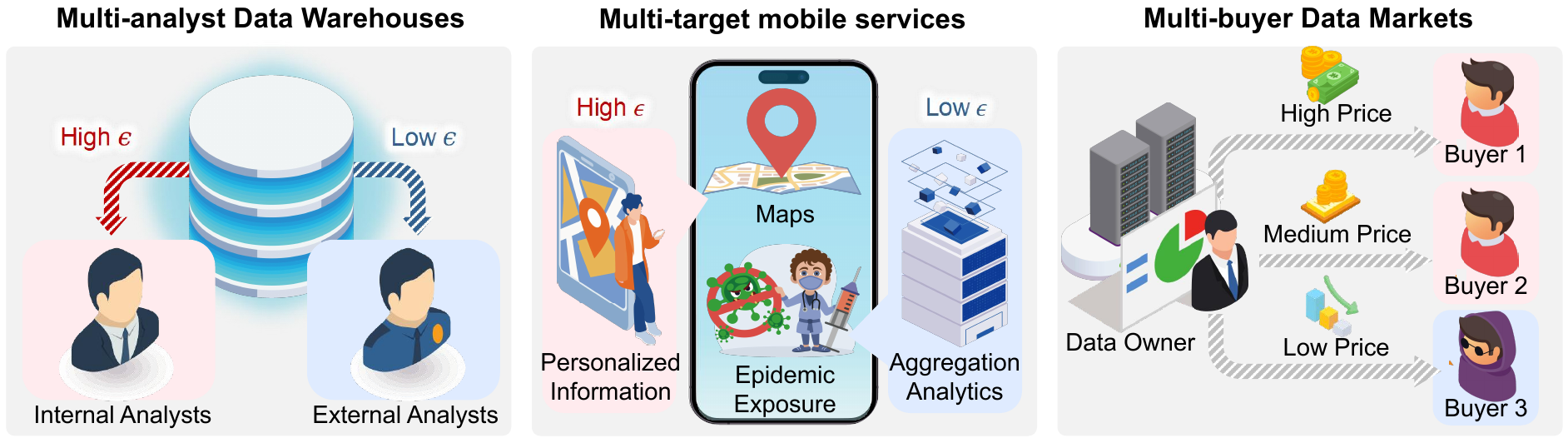}
\vspace*{-1em}
\caption{Application scenarios of multi-tier differentially private queries.}
\label{fig:scenario}
\end{figure*}

A straightforward approach to handling these multi-tier DP queries is to independently invoke a base DP mechanism for each recipient $i\in [m]$ with  privacy budget $\epsilon_i$, to obtain the result $r_i$. We call this \emph{independent release} approach. However, it leads to substantial cumulative privacy loss for the data owner (typically $\sum_{i\in [m]}\epsilon_i$), as the recipients may collude. Specifically, multiple analysts in a data warehouse or multiple buyers in a data market could share and fuse their results to obtain a more accurate version that surpasses each individual privacy level; multiple mobile services (e.g., those from the same parent company) could ensemble several sanitized user reports to achieve an unexpected precision. 

Another approach is to release the responses gradually by invoking a base DP mechanism for each recipient $i\in [m]$ with privacy budget $\epsilon_i-\epsilon_{i+1}$ to obtain $r'_{i}$, assuming that party $i+1$ has a lower budget than party $i$ (i.e. $\epsilon_{i+1} < \epsilon_{i}$) and already receives $r_{i+1}$. The result $r_i$ for the $i$-th recipient can then be obtained by combining $r'_{i}$ and $r_{i+1}$ (e.g., via weighted average). As a consequence, the cumulative privacy loss would be $\epsilon_1$ (i.e., the largest budget among all recipients). We call this the \emph{gradual release} approach. However, combining two or more DP results with split budgets often results in suboptimal utility. For example, averaging $m$ results, each of which is $(\epsilon/m)$-DP with the Laplace mechanism \cite{dwork2006differential} under query sensitivity $1$, yields a variance of ${2m}/{\epsilon^2}$, whereas directly invoking an $\epsilon$-DP Laplace mechanism gives a variance of ${2}/{\epsilon^2}$.

Ideally, solutions for multi-tier DP queries should simultaneously achieve the following two objectives:

 
\begin{itemize}
[leftmargin=1em]
\item\textbf{Privacy objective.} From the data owner's  perspective, the cumulative privacy loss should be bounded by $\max_{i\in [m]} \epsilon_i$. More generally, for any colluding set of parties $S\subseteq [m]$, the joint output $\{r_{i}\}_{i\in S}$ should satisfy $(\max_{i\in S} \epsilon_i)$-DP.
\item\textbf{Utility objective.} From the recipients' perspective, the utility of each received result $r_i$ should (closely) match the optimal/best utility achievable under $\epsilon_i$-DP. 
\end{itemize}

Clearly, the independent release approach fails to achieve the desired privacy goal, whereas the gradual release approach satisfies the privacy objective but often incurs substantial utility loss. Several other attempts have been made to address multi-tier DP queries. Specifically, for DP queries with Gaussian mechanism, because Gaussian noises are closed under linear operations, one can gradually add Gaussian noises to transform a low-privacy result into a high-privacy result \cite{zhang2023dprovdb}. For queries with Laplace mechanism, the work of \cite{koufogiannis2016gradual} also designed a gradual noise addition rule. In the local DP setting, the randomized response mechanism \cite{warner1965randomized,kairouz2016discrete} has simple re-sampling rules to transform a low-privacy result into a high-privacy result \cite{xiao2009optimal}. Nevertheless, these classical base DP mechanisms are known to be suboptimal. For example, in the low-privacy regime, the Laplace mechanism exhibits an exponential error gap compared to the two-sided Geometric mechanism \cite{ghosh2012universally}; meanwhile, the binary randomized response \cite{warner1965randomized} and the general randomized response \cite{kairouz2016discrete}  perform significantly worse than the Subset mechanism \cite{wang2019local} in the low-privacy and high-privacy regimes, respectively. In multi-recipient scenarios, the target privacy levels $\{\epsilon_i\}_{i\in [m]}$ can span from low regimes (e.g., $\epsilon > 1$) to high regimes (e.g., $\epsilon < 1$), making these approaches inherently incapable of achieving the utility objective. Designing DP mechanisms for multi-tier queries that achieve optimal utility remains a fundamental challenge:
\begin{itemize}
[leftmargin=1em]
\item Some state-of-the-art DP noise-adding mechanisms, such as the two-sided Geometric mechanism \cite{ghosh2012universally} and the discrete Gaussian mechanism \cite{canonne2020discrete}, operate in the non-smooth discrete domain (to improve utility and avoid floating-point attacks \cite{mironov2012significance}). \revise{This makes the existence of outputting rules (e.g., via Markov gradual noise adding \cite{koufogiannis2016gradual} for low-to-high privacy transformation) that satisfy the privacy objective of the joint output $\{r_{i}\}_{i\in S}$ difficult to verify, let alone the concrete design of  transformation strategies.}
\item Other DP mechanisms that do not rely on explicit noise addition, such as the Subset mechanism \cite{wang2019local} (used in Google Gboard's OOV discovery \cite{GoogleDP}) and local hash mechanism \cite{wang2017locally}, employ a general probabilistic mapping between the input and output domains, where the output domain differs from the input domain and varies with the privacy budget $\epsilon$. This variation renders privacy adaptation considerably difficult.
\end{itemize}

In this work, we propose new approaches to achieve both goals for the above \emph{multi-tier differentially private query release} problem. Concretely, for noise-adding DP mechanisms (e.g., the two-sided Geometric mechanism), we propose a characteristic function-based framework to check the existence of the privacy transformation and to seamlessly transform a low-privacy result into a high-privacy result (if such a transformation exists), and derive concrete residual noises for the transformation. For other DP mechanisms (e.g., the Subset mechanism in the local DP), we propose a template-based method, which first generates several template results using reference privacy budgets with mechanism-specific operators, and then finds the closest template for each target budget $\epsilon_i$. Through theoretical analyses, we ensure that these methods provide exactly/approximately optimal results for all possible budgets on count queries in the curator and local DP settings, respectively.

The contributions of this work are as follows:
\begin{itemize}[leftmargin=2em,topsep=2pt,itemsep=1pt]
    \item \emph{Multi-tier noise-adding DP framework.} For noise-adding DP mechanisms, we propose a framework that transforms a low-privacy output into a stricter one by adding extra noise, determined by the characteristic function of the noise distributions. We instantiate it for optimal count query mechanisms. Additionally, we demonstrate non-existence of such transformations for the discrete Gaussian mechanism.
    \item \emph{Multi-tier mapping-based DP Method.} For mapping-based DP mechanisms (e.g., the Subset/local hash mechanism), we present a template-based method that generates collusion-resilient tiered templates. We prove the  optimality of each tier's result for the multi-tier Subset mechanism.
    \item \emph{Experimental evaluations. } Through extensive simulations, we demonstrate the effectiveness of our proposed multi-tier frameworks and concrete mechanisms, show significant error reduction compared to existing approaches.
\end{itemize}

The remainder of this paper is organized as follows. Section \ref{sec:pre} provides preliminaries. Section \ref{sec:formulation} formalizes the multi-tier DP query problem. Section \ref{sec:noise} presents the multi-tier framework for noise-adding mechanisms, and Section \ref{sec:transition} develops multi-tier LDP mapping-based mechanisms. Section \ref{sec:exp} presents evaluation results.  Section \ref{sec:related} summarizes related works. Section \ref{sec:conclusion} concludes the paper.


\section{Preliminaries}\label{sec:pre}
\noindent\textbf{Notations. } We use $[c]$ to denote $\{1,2,\ldots,c\}$, and $[c_1,c_2]$ to denote $\{c_1,c_1+1,\ldots,c_2\}$.

\subsection{Characteristic Function}\label{subsec:char-func}
For a random variable $X$, the characteristic function $\Phi_X: \mathbb{R} \mapsto \mathbb{C}$ is defined as:
\[\Phi_X(t) = \mathbb{E}[e^{j t X}],\]
where $j$ is the imaginary unit and $t$ is the parameter of the characteristic function. It captures all moments of the distribution in a single complex-valued function.

The key properties of characteristic functions are as follows:
\begin{itemize}[leftmargin=2em,topsep=2pt,itemsep=1pt]
    \item \emph{Existence property}: The characteristic function $\Phi_X(t)$ exists for every valid probability distribution of $X$.
    \item \emph{Bijection property}: The mapping from a probability distribution to its characteristic function is bijective. A valid characteristic function $\Phi_X$ uniquely determines the distribution of $X$.
    \item \emph{Summation property}: For independent variables $X$ and $Y$,
    \[\Phi_{X+Y}(t) = \Phi_X(t) \cdot \Phi_Y(t).\]
    This extends to finite sums of independent random variables: $\Phi_{\sum_{i=1}^n X_i}(t) = \prod_{i=1}^n \Phi_{X_i}(t)$.
    \item \emph{Mixture property}: For a mixture distribution where $X$ takes values from a family of variables $X_i$ with probabilities $p_i$,
    \[\Phi_X(t) = \sum_i p_i \cdot \Phi_{X_i}(t).\]
\end{itemize}
There exist established criteria to check whether a function is a valid characteristic function and how to numerically derive the probability distribution from a characteristic function \cite{bochner2005harmonic}. These properties of characteristic functions together will facilitate the design of the multi-tier DP noise-adding framework.






\subsection{Differential Privacy}\label{subsec:dp}
\begin{definition}[Hockey-Stick Divergence]\label{def:hs-div}
For distributions $P, Q$ over domain $\mathcal{F}$ and $\alpha \geq 1$, the \emph{$\alpha$-Hockey-Stick divergence} between $P$ and $Q$ is:
$$H_\alpha(P \| Q) = \sum_{v \in \mathcal{F}} \left[ P(v) - \alpha Q(v) \right]_+,$$ where $[x]_+ = \max(x, 0)$.
\end{definition}

Differential privacy (DP) guarantees that the output of an algorithm does not reveal too much information about any individual:
\begin{definition}[$(\epsilon, \delta)$-DP]\label{def:dp}
A mechanism $\mathcal{M}: \mathcal{X}^n \to \mathcal{Y}$ satisfies $(\epsilon, \delta)$-differential privacy if, for any adjacent datasets $X,X'\in \mathcal{X}^n$ that differ by at most one record, the following holds:
\[H_{e^\epsilon}(\mathcal{M}(X) \| \mathcal{M}(X')) \leq \delta.\]
This is equivalent to saying that for any $S \subseteq \mathcal{Y}$:
\[\Pr[\mathcal{M}(X) \in S] \leq e^{\epsilon} \Pr[\mathcal{M}(X') \in S] + \delta.\]
\end{definition}
Typically, $\delta$ is restricted to a small value (e.g., $\delta < O(1/n)$). The parameter $\epsilon$ determines the strength of the privacy guarantee: smaller values mean stronger privacy. The special case $(\epsilon,0)$-DP is often referred to as  $\epsilon$-DP.

Similarly, when each data owner applies privacy mechanisms locally, the definition of the local model of DP is as follows: 
\begin{definition}[$\epsilon$-LDP]\label{def:ldp}
A mechanism $\mathcal{M}: \mathcal{X} \to \mathcal{Y}$ satisfies $\epsilon$-local differential privacy iff for all $x, x' \in \mathcal{X}$:
\[
H_{e^\epsilon}(\mathcal{M}(x) \| \mathcal{M}(x'))=0.
\]
\end{definition}

\textbf{Data processing inequality. } For $f$-divergences, including the Hockey-stick divergence, a desirable property is data processing inequality. It formalizes the fundamental intuition that data processing can only reduce the distinguishability between distributions.
\begin{lemma}[Data Processing Inequality~\cite{polyanskiy2022IT}]\label{thm:dpi}
For any Markov kernel $K$ and $f$-divergence $D$:
\[D(P \| Q) \geq D(K(P) \| K(Q)).\]
\end{lemma}

\textbf{Base noise-adding mechanisms:}
DP is typically achieved by adding carefully calibrated noise to the output of a query. Here are some popular noise-adding mechanisms:
\begin{itemize}[leftmargin=2em, topsep=2pt,itemsep=0pt]
\item \emph{Laplace Mechanism:} Designed for functions $f: \mathcal{D} \to \mathbb{R}$ with sensitivity $\Delta$ (the maximum change in output when one record altered), this mechanism satisfies $\epsilon$-DP by adding noise $\eta$ following a Laplace distribution:
    \[
    \Pr[\eta = y] = \frac{\epsilon}{2\Delta} e^{-\epsilon |y| / \Delta}, \quad \forall y \in \mathbb{R}
    \]
\item \emph{Two-sided Geometric Mechanism~\cite{ghosh2012universally}:} Designed for integer-valued functions $f: \mathcal{X}^n \to \mathbb{Z}$ with sensitivity $\Delta$, this mechanism adds noise $\eta$ following  distribution:
    \[\Pr[\eta = k] = \frac{1 - e^{-\epsilon/\Delta}}{1 + e^{-\epsilon/\Delta}} e^{-\epsilon |k| / \Delta}, \quad \forall k \in \mathbb{Z}\]
    It satisfies $\epsilon$-DP and is universally optimal when $\Delta=1$~\cite{ghosh2012universally}.
\item \emph{Discrete Gaussian Mechanism~\cite{canonne2020discrete}:} Designed for integer-valued functions $f: \mathcal{X}^n \to \mathbb{Z}$, this mechanism adds noise $\eta$ drawn from the discrete Gaussian distribution $N_{\mathbb{Z}}(0,\sigma^2)$:
    \[
    \Pr[\eta = k] = \frac{e^{-k^2/(2\sigma^2)}}{\sum_{j\in\mathbb{Z}} e^{-j^2/(2\sigma^2)}}, \quad \forall k \in \mathbb{Z}.
    \]
    For any $\epsilon,\delta \in (0,1)$ with $\epsilon < 1$, use $\sigma^2 = 2\Delta^2 \log(2/\delta)/\epsilon^2$ ensures $(\epsilon,\delta)$-DP. This mechanism is useful crypto-assisted decentralized DP noise generation.
\item \emph{Skellam Mechanism~\cite{agarwal2021skellam}:} Designed for integer-valued functions $f: \mathcal{X}^n \to \mathbb{Z}$ with sensitivity $\Delta$, this mechanism adds noise $\eta$ following the Skellam distribution:
    \[
    \Pr[\eta = k] = e^{-2\lambda} I_{|k|}(2\lambda), \quad \forall k \in \mathbb{Z},
    \]
    where $I_\nu$ is the modified Bessel function of the first kind. Choosing $\lambda = \Delta/(e^{\epsilon} - 1)$ yields $\epsilon$-DP. The Skellam mechanism is particularly attractive for integer-valued queries in decentralized settings due to its infinite divisibility.

\end{itemize}

\textbf{Base mapping mechanisms:} Besides adding noises, DP can also be realized by explicitly defining a transition/mapping probability matrix between the input and output (e.g., randomized response \cite{warner1965randomized,kairouz2016discrete}, local hash \cite{wang2017locally} and Subset mechanism \cite{wang2019local}).


\begin{itemize}[topsep=2pt,itemsep=0pt]
    \item \emph{Generalized Randomized Response (GRR)}~\cite{kairouz2016discrete}: For categorical data $x \in [d]$ with $\epsilon$-LDP, the output distribution is:
    \[\Pr[y | x] = \begin{cases}
        \frac{e^\epsilon}{e^\epsilon + d - 1} & \text{if } y = x; \\
        \frac{1}{e^\epsilon + d - 1} & \text{if } y \neq x.
    \end{cases}\]

    \item \emph{Local Hash Mechanism}~\cite{wang2017locally}: It maps the domain $[d]$ to a domain $[g]$ using a hash function $H: [d] \to [g]$, then applies GRR over $[g]$. The output distribution is:
    \[\Pr[y | x] = \begin{cases}
        \frac{e^\epsilon}{e^\epsilon + g - 1} & \text{if } y = H(x); \\
        \frac{1}{e^\epsilon + g - 1} & \text{if } y \neq H(x).
    \end{cases}\]
    The hash domain size $g$ is typically set to $\lfloor e^\epsilon + 1\rfloor$.
    
    \item \emph{Subset Mechanism}~\cite{wang2016mutual,wang2019local}: For categorical data $x \in [d]$, it outputs a subset $S \subseteq [d]$ of size $k$, assigns probabilities to subsets based on whether they contain the true value $x$:
    \[\Pr[S | x] = \begin{cases}
        \frac{e^\epsilon}{\binom{d-1}{k-1}e^\epsilon + \binom{d-1}{k}} & \text{if } x \in S; \\
        \frac{1}{\binom{d-1}{k-1}e^\epsilon + \binom{d-1}{k}} & \text{if } x \notin S.
    \end{cases}\]
    The subset size $k$ is set around
    $\frac{d}{1 + e^\epsilon}$. It is optimal in all privacy regimes for discrete distribution estimation \cite{ye2018optimal}.

\end{itemize}

\section{Problem Formulation}\label{sec:formulation}
In this section, we define the multi-tier differentially private query release problem. Without loss of generality to both central DP and local DP (i.e., where a single user's data constitutes the dataset), we consider a database query setting with the following components:
\begin{itemize}
    \item A dataset $X \in \mathcal{D}$, where $\mathcal{D}$ is the domain of the dataset.
    \item A query function $q: \mathcal{D} \mapsto \mathcal{Y}$, which maps a dataset to an output in some domain $\mathcal{Y}$.
    \item A set of analysts $\mathcal{A} = \{1, 2, \dots, m\}$.
    \item A list of privacy budgets $L = \{\epsilon_1, \epsilon_2, \dots, \epsilon_m\}$, where $\epsilon_i > 0$ is the privacy budget assigned to analyst $i \in \mathcal{A}$.
    \item A mechanism $\mathcal{M}: \mathcal{D} \times \mathbb{R}^m \mapsto \mathcal{Y}^m$ that collectively generates $m$ query results $R = \{r_1, r_2, \dots, r_m\}$ given the input dataset and the budget list $L$, where $r_i$ is the result released to analyst $i$ with associated privacy budget $\epsilon_i$.
\end{itemize}

\textbf{Privacy objectives:} We require that the marginal distribution of each single result $r_i$ satisfies the $\epsilon_i$ level of DP (or LDP). Additionally, to limit the accumulated privacy loss from multiple results (e.g., due to the aforementioned collusion scenarios), we also require that for any subset of results, the privacy loss does not exceed the maximum marginal privacy budget among the component results. We call this \emph{collusion resistance} (see Definition \ref{def:collusionresistant}). 

\begin{definition}[Collusion-resistant Property]\label{def:collusionresistant}
A multi-tier mechanism $\mathcal{M}$ satisfies the collusion-resistant property if for any subset $S \subseteq \mathcal{A}$ of analysts, the joint distribution of $\mathcal{M}_S(X,L)$ over the coordinates in $S$ satisfies $(\max_{i \in S} \epsilon_i)$-differential privacy.
\end{definition}

Because the marginal privacy budget of some $r_i$ may be looser (e.g., $r_i$ actually satisfies a DP level $\epsilon'_i$ where $\epsilon'_i < \epsilon_i$), the joint distribution of $\mathcal{M}_S(X,L)$ in the previous collusion resistance condition could exceed $\max_{i\in S} \epsilon'_i$. It is therefore necessary to define a stricter version of collusion resistance. Following the principle of the data processing inequality, we define the \emph{advanced collusion-resistant} property (see Definition \ref{def:collusionresistant2}). Since the hockey-stick divergence satisfies the data processing inequality, it is easy to see that the advanced collusion-resistant property implies the basic collusion-resistant property. 

\begin{definition}[Advanced Collusion-resistant Property]\label{def:collusionresistant2}
A mechanism $\mathcal{M}$ satisfies the advanced collusion-resistant property if, for any subset $S \subseteq \mathcal{A}$ of analysts, any adjacent datasets $X, X' \in \mathcal{D}$, and any distance measure $D$ that satisfies the data processing inequality, the following holds:
\begin{alignat*}{2}
&D(\mathcal{M}_S(X, L) \| \mathcal{M}_S(X', L))\\
\leq&\max_{i\in S} D(\mathcal{M}_{\{i\}}(X, L) \| \mathcal{M}_{\{i\}}(X', L)).
\end{alignat*}
\end{definition}

This definition is equivalent to the \emph{Blackwell ordering} property~\cite{blackwell1953equivalent}: for every subset $S \subseteq \mathcal{A}$ of analysts and every pair of adjacent datasets $X, X' \in \mathcal{D}$, there exists an analyst $i \in S$ (which may depend on $S$, $X$, and $X'$) such that the joint output distribution $\mathcal{M}_S(X, L)$ is a post-processing of $\mathcal{M}_{\{i\}}(X, L)$ and, simultaneously, $\mathcal{M}_S(X', L)$ is the same post-processing of $\mathcal{M}_{\{i\}}(X', L)$.

\textbf{Utility objectives:} From an honest analyst's perspective, we require that the received result $r_i$ maintains optimal utility under the constraint of $\epsilon_i$-DP (or LDP), up to a small constant factor. That is, each analyst $i$ attains utility comparable to what they would obtain from an optimal single-tier mechanism operating at privacy level $\epsilon_i$, thereby. We formalize this approximate optimal utility requirement of the multi-tier DP query system as the \emph{utility-match} property (see Definition \ref{def:utilitymatch}).

\begin{definition}[$\alpha$-Utility-Match Property]\label{def:utilitymatch}
Given an error measure $\text{Error}: \mathcal{Y} \times \mathcal{Y} \mapsto \mathbb{R}$, a positive constant $\alpha \in \mathbb{R}^{\geq 1}$, and a multi-tier mechanism $\mathcal{M}$, we say that $\mathcal{M}$ satisfies the $\alpha$-utility-match property if, for any privacy budget $\epsilon' \in \mathbb{R}^+$, any input dataset $X \in \mathcal{D}$, any $m \in \mathbb{N}^+$, and any budget list $L \in \mathbb{R}^m$ such that $\epsilon'$ is the $i'$-th element of $L$, the result $r_{i'} = \mathcal{M}_{\{i'\}}(X, L)$ achieves error bounds that match the optimal bounds:
\[\mathbb{E}[\text{Error}(r_{i'}, q(X))] \leq \alpha \cdot \mathbb{E}[\text{Error}(\mathcal{M}^*(X, \epsilon'), q(X))],\]
where $\mathcal{M}^*: \mathcal{X} \times \mathbb{R}^+ \mapsto \mathcal{Y}$ denotes the optimal mechanism for query $q$ under privacy budget $\epsilon'$.
\end{definition}
This property provides an incentive for analysts to participate in the multi-tier system.

\textbf{Remark 3.1.} We note that other types of requirements may also be relevant. For example, $\mathcal{M}$ should operate efficiently, in time polynomial in the dataset size and number of analysts. In this paper, our proposals satisfy this efficiency requirement.


\section{Multi-tier Noise-adding Mechanisms}\label{sec:noise}
This section presents a general multi-tier framework for noise-adding DP mechanisms based on characteristic functions, and then provides concrete realizations for count queries.

\subsection{A Framework via Characteristic Functions}
The basic idea of the framework is to gradually inject calibrated noises. It first generates a query result with the highest budget, then uses that high-budget result to generate lower-budget results sequentially, without accessing the dataset again. Therefore, by the data processing inequality, these multi-tier results naturally satisfy the advanced collusion-resistant property. Specifically, our framework leverages the characteristic function (CF) of the noise distribution to derive the calibrated residual noise needed to transform a high-budget result into a low-budget result.

Consider two noise distributions $\eta(\epsilon_i)$ and $\eta(\epsilon_{i+1})$ with $\epsilon_i > \epsilon_{i+1}$, corresponding to some noise-adding DP mechanisms that provide $\epsilon_i$-DP and $\epsilon_{i+1}$-DP guarantees, respectively. Our objective is to transform an output $r_i = q(X) + \eta(\epsilon_i)$ into an output $r_{i+1} \overset{d}{=} q(X) + \eta(\epsilon_{i+1})$ without direct access to $q(X)$ or to the dataset $X$.

The key insight is to exploit the summation property of characteristic functions. For independent random variables, the characteristic function of their sum is the product of their individual characteristic functions. Applying this property, we have:
\[\Phi_{q(X) + \eta(\epsilon_{i+1})}(t) = e^{j t q(X)} \cdot \Phi_{\eta(\epsilon_{i+1})}(t);\]
\[\Phi_{q(X) + \eta(\epsilon_i)}(t) = e^{j t q(X)} \cdot \Phi_{\eta(\epsilon_i)}(t).\]
Dividing these equations isolates the residual noise (denoted as $s_i$):
\[\Phi_{s_i}(t) = \frac{\Phi_{q(X) + \eta(\epsilon_{i+1})}(t)}{\Phi_{q(X) + \eta(\epsilon_i)}(t)} = \frac{\Phi_{\eta(\epsilon_{i+1})}(t)}{\Phi_{\eta(\epsilon_i)}(t)},\]
which represents the characteristic function of the incremental residual noise required to transform $r_i$ into $r_{i+1}$.

Crucially, $\Phi_{s_i}(t)$ depends \emph{only} on the noise distributions of $\eta(\epsilon_i)$ and $\eta(\epsilon_{i+1})$, and is independent of both $q(X)$ and the underlying dataset $X$. This property enables the Markov transformation without direct access to the true query answer $q(X)$ or the dataset $D$. Given an observed output $r_i = q(X) + \eta(\epsilon_i)$, applying the transformation $r_{i+1} = r_i + s_i$ produces a new output distributed identically to $q(X) + \eta(\epsilon_{i+1})$. We present the overall procedure in Algorithm \ref{alg:multi-tier-framework}. The distribution of $s_i$ can be derived from its characteristic function $\Phi_{s_i}(t)$ via the inverse Fourier transform \cite{shephard1991characteristic}, provided that $\Phi_{s_i}(t)$ is a valid characteristic function (a condition that can be checked using established criteria \cite{bochner2005harmonic}).

This framework satisfies the advanced collusion-resistant property, since for any $S \subseteq [m]$, all low-budget results are derived from the highest-budget result in $S$. See formal statement in Theorem \ref{the:noiseaddingprivacy} (proved in Appendix \ref{app:noiseaddingprivacy}). \revise{As shown in the next section, the algorithm may fail depending on the choice of the base mechanism. However, for the most commonly used mechanisms, including the Laplace and Gaussian mechanisms, the algorithm is guaranteed to succeed.}

\begin{figure}[t]
\centering
\includegraphics[width=0.75\linewidth]{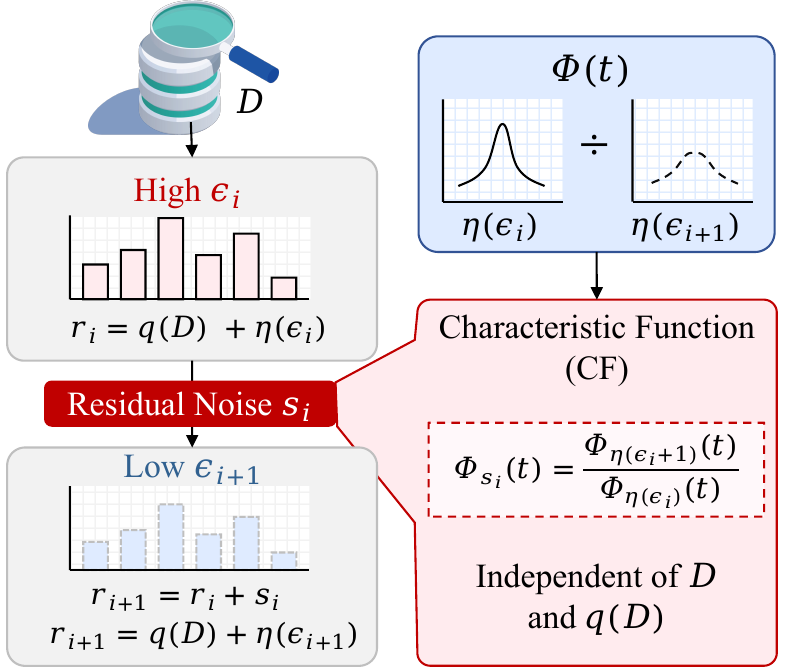}
\caption{Procedures of multi-tier DP noises adding from $\epsilon_i$ to $\epsilon_{i+1}$.}
\label{fig:noiseadding}
\end{figure}

\begin{theorem}[Advanced Collusion Resistance of Multi-tier Noise-adding Framework]\label{the:noiseaddingprivacy}
Algorithm \ref{alg:multi-tier-framework} satisfies the advanced collusion-resistance property.
\end{theorem}



\begin{algorithm}
\caption{Multi-tier DP Noise Adding}
\label{alg:multi-tier-framework}
\KwIn{Dataset $X$, query $q$,  privacy list $\epsilon_1 > \epsilon_2 > \dots > \epsilon_m$, noise distributions $\{\eta(\epsilon_i)\}_{i\in [m]}$}
\KwOut{Outputs $\{r_i\}_{i\in [m]}$ where $r_i \overset{d}{=} q(X) + \eta(\epsilon_i)$}

\tcp{\color{gray}step 1: Generate highest-budget output}
$r_1 \gets q(D) + \eta(\epsilon_1)$

\For{$i = 1$ to $m-1$}{
    \tcp{\color{gray}step 2: Add residual noise}
    $\Phi_{s_i} \gets \frac{\Phi_{\eta(\epsilon_{i+1})}}{\Phi_{\eta(\epsilon_i)}}$

    \eIf{$\Phi_{s_i}$ is a valid characteristic function (CF)}{
        \tcp{\color{gray}derive the residual distribution}
        $s_i\gets \textsf{CF}^{-1}(\Phi_{s_i})$
        
        $s \sim s_i$
        
        $r_{i+1} \gets r_i + s$
    }{
        \Return Failed
    }
}
\Return $\{r_1, r_2, \dots, r_m\}$
\end{algorithm}

\begin{table*}[t]
    \centering
    \setlength{\tabcolsep}{3pt}
    \renewcommand{\arraystretch}{1}
    \begin{tabular}{l|c|c|c|c}
        \hline
       \bfseries Noise Mechanisms  &  \bfseries Distribution &  \bfseries CF & \bfseries Residual CF &  \bfseries Residual Distribution  \\\hline
       
       $Laplace(\epsilon)$ \cite{dwork2006differential} 
       & $\frac{\epsilon}{2}e^{-\epsilon|x|}$ 
       & $\frac{\epsilon^2}{\epsilon^2+t^2}$ 
       & $\frac{\epsilon_{i+1}^2(\epsilon_i^2+t^2)}{\epsilon_i^2(\epsilon_{i+1}^2+t^2)}$ 
       & $\left\{\begin{array}{@{}ll@{}} 0, & \text{w.p. } \frac{\epsilon_{i+1}^2}{\epsilon_i^2} \\ Laplace(\epsilon_{i+1}), & \text{w.p. } 1-\frac{\epsilon_{i+1}^2}{\epsilon_i^2} \end{array}\right.$ \\\hline
       
       $Gaussian(\sigma^2)$ \cite{dwork2008differential} 
       & $\frac{1}{\sqrt{2\pi\sigma^2}}e^{-\frac{x^2}{2\sigma^2}}$ 
       & $e^{-\frac{1}{2}\sigma^2 t^2}$ 
       & $e^{-\frac{1}{2}(\sigma_{i+1}^2-\sigma_i^2) t^2}$ 
       & $Gaussian(\sigma_{i+1}^2 - \sigma_i^2)$ \\\hline
       
       Two-sided Geometric ($TSG(p)$)  \cite{ghosh2012universally} 
       & $\frac{1-p}{1+p}p^{|x|}$ 
       & $\frac{(1-p)^2}{1-2p\cos{t}+p^2}$ 
       & $\frac{(1-p_{i+1})^2(1-2p_{i}\cos{t}+p_{i}^2)}{(1-p_{i})^2(1-2p_{i+1}\cos{t}+p_{i+1}^2)}$ 
       & $\left\{\begin{array}{@{}ll@{}} 0, & \text{w.p.} \frac{(1-p_{i+1})^2p_{i}}{(1-p_{i})^2p_{i+1}} \\ TSG(p_{i+1}),&\text{w.p.} \frac{(p_{i+1}-p_i)(1-p_ip_{i+1})}{(1-p_i)^2 p_{i+1}} \end{array}\right.$ \\\hline
       
       $Skellam(\lambda)$ \cite{agarwal2021skellam} 
       & $e^{-2\lambda}I_{|x|}(2\lambda)$ 
       & $e^{2\lambda(\cos{t}-1)}$ 
       & $e^{2(\lambda_{i+1}-\lambda_i)(\cos{t}-1)}$ 
       & $Skellam(\lambda_{i+1}-\lambda_i)$ \\\hline
        \hline
       
       discrete Gaussian ($DG(\sigma^2)$) \cite{canonne2020discrete} 
       & $\frac{e^{-{x^2}/(2\sigma^2)}}{\sum_{z\in\mathbb{Z}} e^{-{z^2}/(2\sigma^2)}}$ 
       & $\frac{\sum_{x\in\mathbb{Z}} e^{-{x^2}/(2\sigma^2) + jtx}}{\sum_{z\in\mathbb{Z}} e^{-{z^2}/(2\sigma^2)}}$ 
       & \textbf{\color{red} invalid} (in general) 
       & \textbf{\color{red} $\mathbf{\times}$} (see Remark 4.1) \\\hline
       
    \end{tabular}
    \caption{Characteristic functions and residuals of commonly-used noise-adding DP mechanisms.}
    \label{tab:residuals}
\end{table*}

\subsection{Use Cases}
In this section, we instantiate the proposed multi-tier noise addition framework with state-of-the-art base mechanisms for count queries, and also demonstrate cases where the base mechanism fails.

\subsubsection{Multi-tier Count Queries with $\Delta=1$}
The two-sided Geometric (TSG) mechanism is known to be universally optimal for count queries \cite{ghosh2012universally} with $\Delta=1$, and its noise's characteristic function is $\frac{(1-p_i)^2}{1 - 2p_i \cos t + p_i^2}$, where $p_i = e^{-\epsilon_i / \Delta f}$. In the multi-tier noise-adding framework, the residual kernel $\Phi_{s_i}(t)$ is then:
\[\Phi_{s_i}(t) = \frac{\Phi_{\eta(\epsilon_{i+1})}(t)}{\Phi_{\eta(\epsilon_i)}(t)} = \frac{(1-p_{i+1})^2 (1 - 2p_i \cos t + p_i^2)}{(1-p_i)^2 (1 - 2p_{i+1} \cos t + p_{i+1}^2)}.\]
By the mixture property, this residual noise's characteristic function can be decomposed as:
\begin{alignat*}{2}
&\frac{(1-p_{i+1})^2 (1 - 2p_i \cos t + p_i^2)}{(1-p_i)^2 (1 - 2p_{i+1} \cos t + p_{i+1}^2)} \\
=&\frac{(1-p_{i+1})^2 p_i}{(1-p_i)^2 p_{i+1}} \cdot 1 + \frac{(p_{i+1} - p_i)(1 - p_i p_{i+1})}{(1-p_i)^2 p_{i+1}} \cdot \frac{(1-p_{i+1})^2}{1 - 2p_{i+1} \cos t + p_{i+1}^2},
\end{alignat*}
where $1$ is the characteristic function of the Dirac distribution at value $0$, and $\frac{(1-p_{i+1})^2}{1 - 2p_{i+1} \cos t + p_{i+1}^2}$ is the characteristic function of a TSG noise. Therefore, the residual noise $s_i$ is as follows:
\begin{equation}\label{eq:tsgresidual}
    s_i = \left\{
    \begin{array}{@{}ll@{}}
        0, & \text{w. p. } \frac{(1-p_{i+1})^2p_{i}}{(1-p_{i})^2p_{i+1}}; \\
        \text{TSG}(p_{i+1}), & \text{w. p. } \frac{(p_{i+1}-p_i)(1-p_ip_{i+1})}{(1-p_i)^2 p_{i+1}}.
    \end{array}
\right.
\end{equation}

\subsubsection{Multi-tier Count Queries with $\Delta > 1$}
We note that an optimal mechanism for count queries with $\Delta > 1$ when $\epsilon > 1$ (i.e., \emph{MSDLap} \cite{harrison2025infinitely}) is available via weighted summation of $\Delta$ independent two-sided Geometric noises, each with budget $\epsilon$. Specifically, it achieves an MSE error rate of $O(\Delta^3 \cdot e^{-\epsilon})$ and is constructed as follows:
\[
\text{MSDLap}(\epsilon) = \sum_{\xi \in [\Delta]} \xi \cdot \text{TSG}(\epsilon).
\]
For multi-tier privacy levels, one can simply add $\Delta$ weighted independent residual noises from Equation \ref{eq:tsgresidual} (i.e., $\sum_{\xi \in [\Delta]} \xi \cdot r_i$) to transform from $\text{MSDLap}(\epsilon_i)$ to $\text{MSDLap}(\epsilon_{i+1})$, thereby obtaining a multi-tier version of the MSDLap mechanism.


\subsubsection{Other Queries. } Many other popular DP noise-adding mechanisms (e.g., Laplace \cite{dwork2006differential}, Gaussian \cite{dwork2008differential}, Skellam \cite{agarwal2021skellam}) can also be turned into multi-tier versions within our framework. We list their characteristic functions and residual noises in Table \ref{tab:residuals}. Specifically, the Laplace mechanism are commonly used for numerical queries without non-integer constraints; the Skellam mechanism are particular useful for decentralized DP queries over cryptographic secure aggregation where query results must be quantized for compatible in the cryptographic domain and each party injects partial noises. \revise{These mechanisms, together with the two-sided geometric mechanism, encompass most use cases of noise-adding DP (including those implemented in popular repositories such as OpenDP, IBM's Diffprivlib, PySyft, Google DP, TensorFlow Privacy, and Opacus).}


\textbf{Remark 4.1.} Our framework also facilitates analyzing the existence of multi-tier versions for other prevalent noise-adding mechanisms. For example, the discrete Gaussian (DG) mechanism \cite{canonne2020discrete} has characteristic function:
\[
\Phi_{\sigma}(t) := \left({\sum\nolimits_{k \in \mathbb{Z}} e^{-k^2/(2\sigma^2)}e^{jkt}}\right)/\left( {\sum\nolimits_{k \in \mathbb{Z}} e^{-k^2/(2\sigma^2)}}\right),
\]
and the residual function is $\Phi_{\sigma_{i+1}}(t) / \Phi_{\sigma_i}(t)$.
Recall that any valid characteristic function $R: \mathbb{R} \to \mathbb{C}$ must be \emph{positive definite} \cite{bochner2005harmonic}: for every $n \in \mathbb{N}$ and every choice of points $t_1, \dots, t_C \in \mathbb{R}$, the associated \emph{Bochner matrix}
\[
M := \bigl[ M_{a,b} \bigr]_{a,b \in [C]}, \qquad
M_{a,b} := R(t_a - t_b),
\]
must be Hermitian positive semi-definite. Equivalently, all eigenvalues of $M$ must be non-negative. We construct a counterexample for the residual function $$R(t) = \Phi_{\sigma_{i+1}}(t) / \Phi_{\sigma_i}(t)$$ with $\sigma_{i+1} = 1.1$ and $\sigma_i = 1$. When $t$ takes values in $\{\pi \cdot (k-1)/2\}_{k \in [4]}$, the Bochner matrix $M \in \mathbb{C}^{4 \times 4}$ with entries $M_{a,b} = R(t_a - t_b)$ is:
$$
\begin{bmatrix}
1 & 0.771729 & 0.354762 & 0.771729 \\
0.771729 & 1 & 0.771729 & 0.354762 \\
0.354762 & 0.771729 & 1 & 0.771729 \\
0.771729 & 0.354762 & 0.771729 & 1 \\
\end{bmatrix}
$$
It has a smallest eigenvalue of approximately $-0.1886953 < 0$.
Therefore, $M$ is not positive semi-definite, the $\Phi_{\sigma_{i+1}}(t) / \Phi_{\sigma_i}(t)$ is not a valid characteristic function, and no additive noise can transform from $\text{DG}(1^2)$ to $\text{DG}(1.1^2)$. \revise{A similar failure happens to the staircase mechanism \cite{geng2015staircase} (see Appendix \ref{app:staircase})}.

\section{Multi-tier Mapping Mechanisms}\label{sec:transition}
In this section, we study multi-tier DP count queries  in the local model of DP, which often employ mapping-based mechanisms. Intuitively, these mechanisms define a general transition or mapping probability matrix, and the output domain may not align with the input domain (unlike noise-adding mechanisms). Particularly,  for count queries over categorical data, we propose a multi-tier LDP mechanism based on the  Subset mechanism \cite{wang2019local}.

\revise{
\subsection{A Template-based Method}
Many state-of-the-art mapping-based mechanisms use output domains with hyperparameters, whose optimal values depend on the privacy budget $\epsilon$. Examples include the subset size in Subset mechanisms \cite{wang2019local}, the hash domain size in local hashing \cite{wang2017locally}, and the partition size in Hadamard response \cite{acharya2019hadamard}. These mechanisms face challenges in multi-tier settings: transitions between privacy levels require adapting both hyperparameters and probability ratios that determine privacy levels, and it might be impossible to meet optimal hyperparameter conditions for all target levels $\epsilon_i$.
As a result, utility gaps may arise, unlike in multi-tier noise-adding mechanisms, where smooth residual noise injection is often possible.

To address this, we propose a template-based method. It firstly defines primitive operators that adapt hyperparameters while maintaining distributional equivalence, then identifies reference privacy levels where these transitions are viable and generate corresponding templates, finally matches the privacy list $L$ to the templates. Specifically, the method works as follows:
\begin{enumerate}
[leftmargin=2em]
\item\emph{Design primitive operators:} For a base mechanism with hyperparameter $\theta$, define operators that can adapt $\theta$  and/or adjust the probability ratio to achieve target privacy levels, while preserving distributional equivalence to the base mechanism.

\item \emph{Set reference privacy levels:} Identify privacy levels $\{\epsilon'_k\}$ in descending order where hyperparameter and probability ratio transitions are naturally supported by the primitive operators. 

\item \emph{Generate templates:} For each $\epsilon'_k$, sequentially generate a template $r'_k$ from $r'_{k-1}$ using the primitive operators. 

\item \emph{Match levels:} For target level $\epsilon_i$, select the closest template with $\epsilon'_k \leq \epsilon_i$.
\end{enumerate}

When reference privacy levels are close to each other, the utility gap between the optimal error bound of $\epsilon_i$ and the matched template are guaranteed to be small. In following subsections, we demonstrate this method using the Subset mechanism \cite{wang2016mutual,wang2019local}. This method also applies to the local hash mechanism \cite{wang2017locally} (see Appendix \ref{app:multitierLH}). We focus on the Subset mechanism, as it is known to achieve optimal results for LDP count queries in all privacy regimes \cite{ye2018optimal}.
}

\subsection{Multi-tier Subset Mechanism}\label{subsec:subset}
Recall that in the Subset mechanism $\mathcal{R}$ over a one-hot data domain $x \in \{0,1\}^d$ with budget $\epsilon$ and subset size $k$, the variance is:
\[
V(\epsilon, k) = \mathbb{E}[(\mathcal{R}(x) - x)_2^2] = \frac{t(1-t) + (d-1)f(1-f)}{(t-f)^2},
\]
where $t = \frac{k e^\epsilon}{k e^\epsilon + d - k}$ and $f = \frac{k e^\epsilon (k-1) + (d-k)k}{(k e^\epsilon + d - k)(d-1)}$. The variance formula can be further simplified as:
\[
V(\epsilon, k) = \frac{(d-1)\left(k - d + (d-k)^2 + 2e^{\epsilon}(d-k)k + e^{2\epsilon}(k-1)k\right)}{(e^{\epsilon} - 1)^2 (d-k)k}.
\]

Theoretical analyses \cite{wang2016mutual,wang2019local,ye2018optimal} indicate that $k \approx d / (e^{\epsilon} + 1)$ minimizes the estimation variance, and the optimal subset size is non-decreasing with the privacy budget. Therefore, we introduce two primitive operations for adjusting the subset size/privacy sequentially (from a relatively lower privacy setting to a relatively higher one) in the multi-tier setting without accessing the original true value:

\textbf{Expansion operator.} Given a Subset output $r$ of size $k$ and privacy level $\epsilon$ (i.e., $r \overset{\mathrm{d}}{=} \text{Subset}(x, k, \epsilon)$), the expansion operator extends the subset by one element. Specifically, we add one uniformly random element from $[d] \setminus r$ to $r$, obtaining:
\[
\text{Expansion}(r) = r \cup \{\text{uniform}([d] \setminus r)\}.
\]
In Lemma \ref{lemma:expansion}, we show that the expanded result is distributionally equivalent to a Subset output with a lower budget.

\begin{lemma}[Properties of expansion operation]\label{lemma:expansion}
Let $r\in [d]^k$ denote the output variable of the subset mechanism tightly satisfies with budget $\epsilon$ (i.e., $r\overset{\mathrm{d}}{=}  Subset(x,k,\epsilon)$), then $$\text{Expansion}(r) \overset{\mathrm{d}}{=}Subset\left(x,k+1,\ln\left(\frac{e^{\epsilon}k+1}{k+1}\right)\right).$$
\end{lemma}

\textbf{Rescale operator.} Given a Subset output $r$ and a parameter $\beta \in [0,1)$, the rescale operation $\text{Rescale}(r, \beta)$ outputs the original input $r$ with probability $\beta$, otherwise it samples a uniform random element from $\{s \subseteq [d] : |s| = |r|\}$ as the output:
$$Rescale(r,\beta)=\begin{cases}
r,\ \ \ \ \ \ \ \ \ \ \ \ \ \ \ \ \ \ \ \ \ \ \ \ \ \ \ \ \ \ \ \ \ \ \ \ \ \ \ \ \ \ \ \ \ \ \ \ w.p.\  \beta;\\
\textsf{uniform}(\{s \subseteq [d] : |s| = |r|\}),\ \ w.p.\  1-\beta.\end{cases}$$
In Lemma \ref{lemma:rescale}, we show that $\text{Rescale}(r, \beta)$ satisfies a new level of LDP and is distributionally equivalent to a new level Subset output (see Appendix \ref{app:rescale} for proof). Conversely, setting the parameter to
\[
\beta = \frac{(e^{\epsilon'} - 1)(k e^{\epsilon} + d - k)}{d(e^{\epsilon} - e^{\epsilon'})},
\]
transforms the privacy level from $\epsilon$ to $\epsilon'$ (assuming $\epsilon' < \epsilon$).

\begin{lemma}[Properties of rescale operation]\label{lemma:rescale}
Let $r\in [d]^k$ denote the output variable of the subset mechanism with budget $\epsilon$ (i.e., $r\overset{\mathrm{d}}{=}  Subset(x,k,\epsilon)$), then for $0\leq \beta\leq 1$: $$\text{Rescale}(r,\beta)\overset{\mathrm{d}}{=}Subset\left(x,k,\log\left(\frac{d\beta e^{\epsilon} + (1-\beta)(k e^{\epsilon} + d - k)}{d\beta + (1-\beta)(k e^{\epsilon} + d - k)}\right)\right).$$ 
\end{lemma}


\textbf{Failure of direct adaptation.} When transforming from a higher budget $\epsilon_i$ to a lower budget $\epsilon_{i+1}$, in order to achieve desirable utility, we may need to adapt both the subset size and the probability ratio (i.e., $e^\epsilon$). However, a straightforward approach that sequentially adapts the subset size and the privacy level for $\epsilon_1, \dots, \epsilon_m$ may be infeasible, since not all exact optimality conditions can be satisfied simultaneously for these privacy levels. For example, with $d=10$, $\epsilon = \ln(6.4)$, and $\epsilon' = \ln(5)$, the optimal subset sizes that minimize variance are $k=1$ and $k'=2$, respectively. After applying an expansion operation to $\text{Subset}(x, k, \epsilon)$, the privacy level of $\text{Expansion}(r)$ becomes $\ln\!\left(\frac{6.4 \cdot 1 + 1}{1 + 1}\right) = \ln(3.7)$, which is lower than the next target privacy level $\epsilon' = \ln(5)$. This also implies that achieving \revise{exactly best utility} for all privacy levels, as in the single-tier subset mechanism, is generally impossible.

\begin{figure*}[t]
\centering
\includegraphics[width=1.0\linewidth]{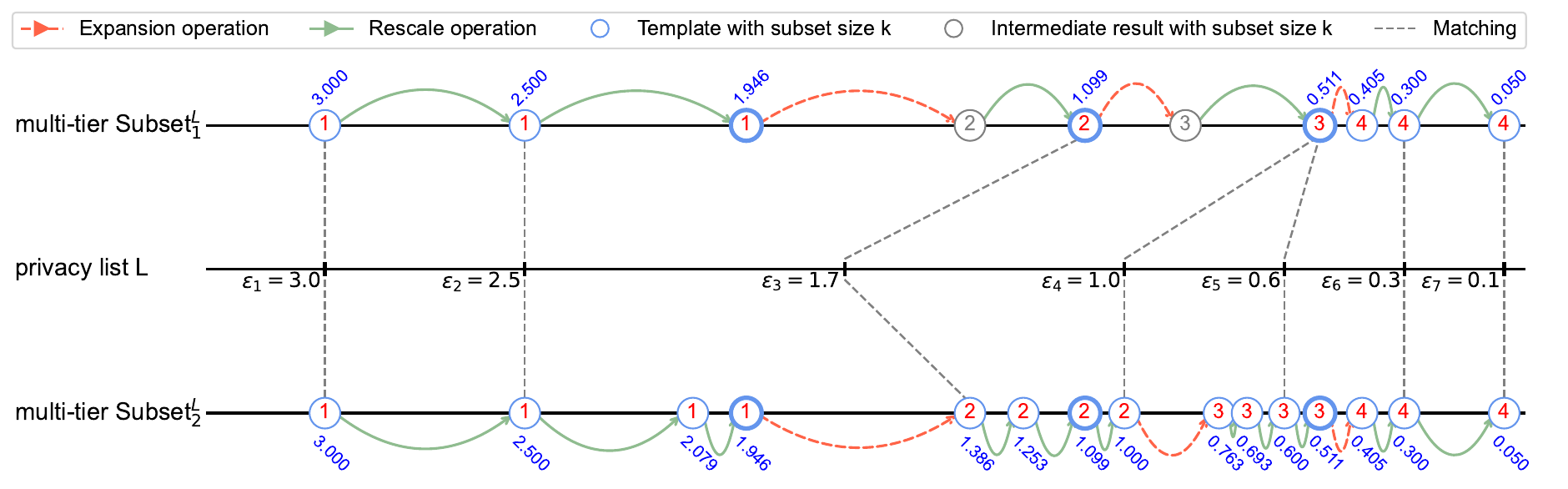}
\vspace*{-1.5em}
\caption{Procedure demonstration of multi-tier subset mechanism with $d=8$.}
\label{fig:multitiersubset}
\end{figure*}

\textbf{A template-based approach.} We aim to \revise{approximately match the best utility and employ the template-based method} that generates a series of viable outputs using the two operations. This series corresponds to a list of reference privacy levels $$\{\ln(d/k - 1)\}_{k \in [1, \lfloor (d-1)/2 \rfloor]},$$ whose optimal subset size changes from $1$ to $\lfloor (d-1)/2 \rfloor$. This is sequentially viable using the \emph{expansion} and \emph{rescale} operators, because when applying the expansion operator to the previous template result $r$ with budget $\ln(d/k - 1)$, according to Lemma \ref{lemma:expansion}, $\text{Expansion}(r)$ satisfies the LDP level of:
\begin{equation}\label{eq:expansionnext}
\ln\left(\frac{k(d/k - 1) + 1}{k + 1}\right) = \ln\left(\frac{d - k + 1}{k + 1}\right) > \ln\left(\frac{d}{k + 1} - 1\right),
\end{equation}
which is always larger than the next reference budget $\ln(d/(k + 1) - 1)$. Finally, we match each target level in $L$ to the closest lower-budget template (see Algorithm \ref{alg:templatematch}). The advantage of this template-based approach over naive sequential adaptation is clear: the template results for the reference privacy levels $\{\ln(d/k - 1)\}_{k \in [1, \lfloor (d-1)/2 \rfloor]}$ are exactly optimal at their levels, and each privacy level $\epsilon_i$ in $L$ can be matched to a nearby template result. Recall the previous example: the privacy level in the naive sequential approach jumped from $\ln(6.4)$ to $\ln(3.7)$, which is lower than the reference privacy level $\ln(4) = \ln(10/2 - 1)$; matching the next target level $\ln(5)$ to $\ln(3.7)$ causes more excess utility loss than matching it to $\ln(4)$. The overall multi-tier subset mechanism (presented in Algorithm \ref{alg:multitiersubset}) consists of five major steps:
\begin{itemize}[leftmargin=1.5em]
    \item[(1)] \emph{Handle low privacy regimes}: For privacy levels $\epsilon_{i+1}$ such that $e^{\epsilon_{i+1}} + 1 > d$ (where $k_{i+1}^* = 1$), sequentially generate the $1$-subset mechanism's output. Use the Rescale operator with 
    \[
    \beta = \frac{(e^{\epsilon_{i+1}} - 1)(e^{\epsilon_i} + d - 1)}{d(e^{\epsilon_i} - e^{\epsilon_{i+1}})}
    \]
    during the sequential procedure, ensuring that the resulting output satisfies the target level of $\epsilon_{i+1}$-LDP (assuming $\epsilon_0 = +\infty$).
    \item[(2)] \emph{Generate middle regime templates}: For each $k \in [1, \lfloor (d-1)/2 \rfloor]$, sequentially generate a template result for $\epsilon'_k = \ln(d/k - 1)$. Each generation iteration involves two steps: first expanding the subset size by $1$, and then rescaling the probability ratio to exactly $d/k - 1$. This is summarized in Algorithm \ref{alg:templategen}, for which Lemma \ref{lemma:expansion} ensures that the resulting probability ratio $\rho$ after expansion is always larger than $d/k - 1$ (i.e., $\epsilon'_k$ is exactly achievable by a further \emph{rescale} operator).
    \item[(3)] \emph{Handle high privacy regimes when $(d+2)/\lfloor d/2 \rfloor < e^{\epsilon} + 1 < d / \lfloor (d-1)/2 \rfloor$}: When $d$ is even, for the privacy regime where $(d+2)/\lfloor d/2 \rfloor < e^{\epsilon} + 1 < d / \lfloor (d-1)/2 \rfloor$, expand the subset size from $\lfloor (d-1)/2 \rfloor$ to $\lfloor d/2 \rfloor$ and save the corresponding result.
    \item[(4)] \emph{Handle high privacy regimes when $e^{\epsilon} + 1 < (d+2)/\lfloor d/2 \rfloor$ (if $d$ is even) or $e^{\epsilon} + 1 < d / \lfloor (d-1)/2 \rfloor$ (if $d$ is odd)}: Sequentially generate the corresponding result from the previous result for the specific $\epsilon_{i+1}$ using only the \emph{rescale} operator.
    \item[(5)] \emph{Match privacy list to results}: Given the generated results $R$, for each target level $\epsilon_i \in L$, find the closest result that satisfies $\epsilon_i$ and set that result as $r_i$.
\end{itemize}

\begin{algorithm}[t]
\caption{Multi-tier Subset Mechanism $Subset^L(x)$}
\label{alg:multitiersubset}
\KwIn{Privacy list $L=\{\epsilon_1,\ldots,\epsilon_m\}$; the true value $x\in [d]$}
\KwOut{Query results $\{r_{i}\}_{i\in [m]}$ each with $\epsilon_{i}$-LDP}

$R\leftarrow\{\}$;\ \ 
$s\leftarrow \{x\}$;\ \ 
$\rho \leftarrow +\infty$\;

\tcp{\color{gray}Step 1: handle cases when $e^\epsilon_i+1> d$}
\For{$\epsilon_i$ in $L$ where $e^{\epsilon_i}+1> d$}
{
$s \leftarrow \rescale\left(s,\frac{(e^{\epsilon_{i}}-1)(|s| \rho + d - |s|)}{d(\rho - e^{\epsilon_{i}})}\right)$;\ \ $\rho \leftarrow e^{\epsilon_i}$\;
append $(\epsilon_i,s)$ to $R$\;
}

\tcp{\color{gray}Step 2: generate middle templates}
$R',s,\rho=$TeampleGen$(s,\rho)$\;
append elements in $R'$ to $R$\;

\tcp{\color{gray}Step 3: handle cases when $d$ is even}
{
    $\rho \leftarrow \frac{\rho|s| + 1}{|s| + 1}$;\ \ $s \leftarrow \expansion(s)$\;
    append $(\ln(\rho),s)$ to $R$\;
}

\tcp{\color{gray}Step 4: cases when $e^{\epsilon_{i}}+1< (d+2)/\lfloor d/2\rfloor$ (if $d$ is even) or $e^{\epsilon_{i}}+1< d/\lfloor (d-1)/2\rfloor$ (if $d$ is odd)}
\For{$\epsilon_i$ in $L$ where $e^{\epsilon_i}<\rho$}
{

$s \leftarrow \rescale\left(s, \frac{(e^{\epsilon_{i}}-1)(|s| \rho + d - |s|)}{d(\rho - e^{\epsilon_{i}})}\right)$;\ \ $\rho \leftarrow e^{\epsilon_i}$\;
append $(\epsilon_i,s)$ to $R$\; 
}

\tcp{\color{gray}Step 5: match privacy list to results}
$\{r_1,\ldots,r_m\}\leftarrow$TemplateMatch$(R,L)$

\Return{$\{r_1,\ldots,r_m\}$}
\end{algorithm}

\begin{algorithm}[t]
\caption{TemplateGen($s,\rho$)}
\label{alg:templategen}
\KwIn{Previous result $s$ satisfying $\ln(\rho)$-LDP and $|s|=1$ with $\rho+1>d$.}
\KwOut{Templates $\{r'_k\}_{k\in[1,\lfloor(d-1)/2\rfloor]}$ each satisfies $\ln(d/k-1)$-LDP.}

$R'\leftarrow\{\}$

\For{$k\in [1, \lfloor (d-1)/2\rfloor]$}{
 \If{$k>|s|$}{
    \tcp{\color{gray}expand the subset size by $1$}
    $\rho \leftarrow \frac{\rho|s| + 1}{|s| + 1}$; $s\leftarrow \expansion(s)$\;   
}
\tcp{\color{gray}rescale the probability ratio/budget}
$\epsilon'_k\leftarrow \ln(d/k-1)$\;
$s \leftarrow \rescale\left(s,\frac{(e^{\epsilon'_{k}}-1)(k \rho + d - k)}{d(\rho - e^{\epsilon'_{k}})}\right)$; $\rho \leftarrow e^{\epsilon'_k}$\;
append $(\epsilon'_k,s)$ to $R'$\;
}
\Return{$R',s,\rho$}
\end{algorithm}

\begin{algorithm}[h]
\caption{TemplateMatch($R,L$)}
\label{alg:templatematch}
\KwIn{Sorted templates $R$; privacy list $L=\{\epsilon_1,\ldots,\epsilon_m\}$}
\KwOut{Results $\{r_i\}_{i\in[1,m]}$ with $r_i$ satisfies $\epsilon_i$-LDP\;}

\tcp{\color{gray}$R[l].\epsilon$: the $l$-th template's level}
\tcp{\color{gray}$R[l].s$: the $l$-th template's result}


\For{$i \leftarrow 1$ \KwTo $m$}{
    $l_i \leftarrow \arg\min_{l\in [|R|]} \{R[l].\epsilon \leq \epsilon_i\}$\;
    $r_i \leftarrow R[l_i].s$\;
}

\Return{$\{r_1,\ldots,r_m\}$}
\end{algorithm}

We depict an example of the procedure in Figure \ref{fig:multitiersubset} (denoted as multi-tier Subset$_1^L$). We note that each template $\{r'_k\}_{k}$ satisfies exactly $\ln(d/k - 1)$-LDP and is exactly optimal at the corresponding privacy level. These templates help provide desirable results (with sufficiently close optimal subset size and close privacy level) in the middle privacy regime. For privacy levels in the low and high privacy regimes, the mechanism exactly matches the optimal subset mechanism. We provide the overall utility-match (see Definition \ref{def:utilitymatch}) guarantee of the multi-tier subset mechanism in Theorem \ref{the:utilitymatch}. Theoretically, the multiplicative factor is upper bounded by 171/11 (see Appendix \ref{app:utilitymatch}). Numerically, its value is relatively small in most cases (e.g., $\leq 3$), according to Figure \ref{fig:heatmap} (left part).

\begin{figure*}[t]
\centering
\includegraphics[width=0.99\linewidth]{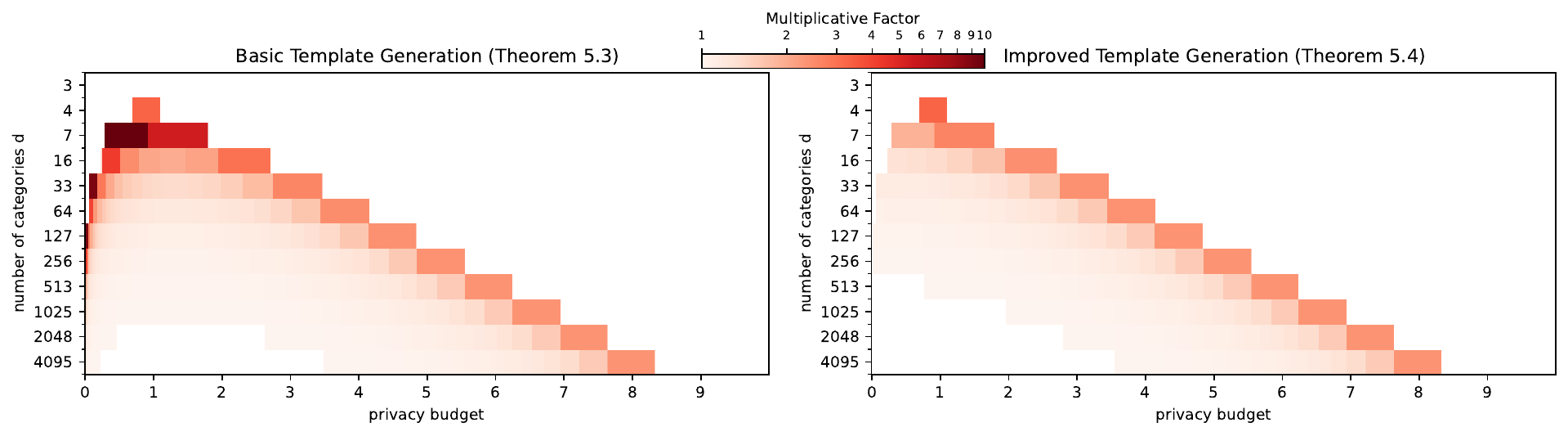}
\vspace*{-1em}
\caption{Heatmap of theoretical multiplicative factor upper bounds of Theorem \ref{the:utilitymatch} and Theorem \ref{the:utilitymatch2}, under varying categories $d$ and privacy budget $\epsilon_i$ (grid values from $0.01$ to $10$ with granularity $0.01$).}
\label{fig:heatmap}
\end{figure*}

\begin{theorem}[Utility of Multi-tier Subset Mechanism]\label{the:utilitymatch}
For any $\epsilon_i\in L$ in the multi-tier subset mechanism $Subset^{L}(x)$, let $k^*_i$ denote the optimal subset size $\arg\min_{k\in [1,\ldots,\lfloor d/2\rfloor]} V(\epsilon_i,k)$, the correspond result $Subset^{L}(x)_i=r_i$ satisfies:
\begin{alignat*}{2}
&\frac{\mathbb{E}[\|Subset^{L}(x)_i-x\|_2^2]}{V(\epsilon_i,k^*_i)}\\
\leq&\max\!\left\{\frac{V\left(\ln\left(\frac{d}{\overline{k}_i}-1\right), \overline{k}_i\right)}{V\left(\ln\left(\frac{d}{\overline{k}_i-1}-1\right), \overline{k}_i-1\right)},\frac{V\left(\ln\left(\frac{d+2}{\lfloor d/2\rfloor}-1\right),\lfloor \frac{d}{2}\rfloor\right)}{V\left(\ln\left(\frac{d}{\lfloor (d-1)/2\rfloor}-1\right),\lfloor \frac{d-1}{2}\rfloor\right)}\right\}.
\end{alignat*}
\end{theorem}

Since the original subset mechanism is optimal in all privacy regimes and the multi-tier subset mechanism approximates the best single-tier subset mechanism up to a constant factor, the multi-tier subset mechanism is order-optimal for all privacy regimes.

\subsection{Improving Template Generation}
The template generation algorithm can be further improved. The improvement exploits the gap in Equation \ref{eq:expansionnext}, where the value of $k \cdot (\rho + 1)$ increases by \emph{two} after each expansion operation. This implies that additional templates can be generated for free within the gap regions. For example, after each expansion operation for $k \in [1, \lfloor (d-1)/2 \rfloor]$, one can further use the rescale operation for: $$\epsilon_i \in \left(\min\{\rho, \ln((d+1)/k - 1)\}, \ln((d-1)/k - 1)\right]$$
to generate extra templates. These extra templates can be closer to the privacy levels in $L$ and thus improve utility over Algorithm \ref{alg:multitiersubset}. We summarize the new template generation process in Algorithm \ref{alg:templategen2}, which takes as input the previous result and the privacy list $L$.

\begin{algorithm}[t]
\caption{TemplateGen2($s,\rho, L$)}
\label{alg:templategen2}
\KwIn{Previous result $s$ satisfying $\ln(\rho)$-LDP and $|s|=1$; a list of privacy budget $L$}
\KwOut{A list of template randomization results}

$L'\leftarrow L$\;

\tcp{\color{gray}add more reference privacy levels}
\For{$k\in [1, \lfloor (d-1)/2\rfloor]$}{
    add $\ln(d/k-1)$, $\ln((d+1)/k-1)$ to $L'$
}

sort $L'$ in descending order

$R'\leftarrow\{\}$

\For{$k\in [1, \lfloor (d-1)/2\rfloor]$}{

     \If{$k>|s|$}{
        $\rho \leftarrow \frac{\rho|s| + 1}{|s| + 1}$; $s \leftarrow \expansion(s)$\;
        append $(\log(\rho),s)$ to $R'$\; 
     }

    \For{$\epsilon_i\in L'$ that $\epsilon_i\in [\ln((d-1)/k-1, \ln(\rho)))$}
    {
        \If{$k\neq\lfloor (d-1)/2\rfloor$ or $\epsilon_i\geq \ln(d/\lfloor (d-1)/2\rfloor-1)$}
        {
        $s\leftarrow \rescale\left(s,\frac{(e^{\epsilon_{i}}-1)(k \rho + d - k)}{d(\rho - e^{\epsilon_{i}}}\right)$\;
        $\rho \leftarrow e^{\epsilon_i}$\;
        append $(\epsilon_i,s)$ to $R'$\; 
        }
    }
}
\Return{$R',s,\rho$}
\end{algorithm}

We present the utility match property of the multi-tier Subset mechanism with improved template generation in Theorem \ref{the:utilitymatch2}. It reduces the multiplicative factor of utility matching guarantee compared to Theorem \ref{the:utilitymatch}. Figure \ref{fig:heatmap} demonstrates the numerical values of the multiplicative factor bounds under various $d$ and $\epsilon_I$, with comparison to Theorem \ref{the:utilitymatch}. In extensive settings, the improved template generation could limit the upper bound of  multiplicative factors in a desirable range (e.g. $<2$).

\begin{theorem}[Utility of Multi-tier Subset Mechanism with Improved Template Generation]\label{the:utilitymatch2}
For any $\epsilon_i\in L$ in the multi-tier subset mechanism $Subset_2^{L}(x)$ with template generation algorithm \ref{alg:templategen2}, let $k^*_i$ denote the optimal subset size $\arg\min_{k\in [1,\ldots,\lfloor d/2\rfloor]} V(\epsilon_i,k)$, $\overline{k}_i$ denote $\max\{2,\lceil d/(e^\epsilon_i+1)\rceil\}$, the result $Subset_2^{L}(x)_i=r_i$ satisfies:
\begin{alignat*}{2}
&\frac{\mathbb{E}[\|Subset_2^{L}(x)_i-x\|_2^2]}{V(\epsilon_i,k^*_i)}\\
\leq& \max\left\{\frac{V\left(\ln\left(\frac{d+1}{\overline{k}_i}-1\right), \overline{k}_i\right)}{V\left(\ln\left(\frac{d-1}{\overline{k}_i-1}-1\right), \overline{k}_i-1\right)},\frac{V\left(\ln\left(\frac{d+2}{\lfloor d/2\rfloor}-1\right),\lfloor \frac{d}{2}\rfloor\right)}{V\left(\ln\left(\frac{d}{\lfloor (d-1)/2\rfloor}-1\right),\lfloor \frac{d-1}{2}\rfloor\right)}\right\}.
\end{alignat*}
\end{theorem}

\textbf{Computation and memory complexity.} For the multi-tier algorithm using either the template generation procedure in Algorithm \ref{alg:templategen} or \ref{alg:templategen2}, the resulting set of templates $R$ contains at most $m + \lfloor d/2 \rfloor$ (sorted) elements. Additionally, the $\textsf{TemplateMatch}$ procedure, which takes two sorted lists as input, has computational and memory complexity of $O(|L| + |R|)$. Therefore, the overall computational and memory complexity of $\textsf{Subset}^L(x)$ is $O(|L| + d)$.

\section{Experiments}\label{sec:exp}
In this section, we present numerical experimental results of our multi-tier DP proposals, with comparison to the following baseline:

\textbf{The gradual release baseline.} Recall the \emph{gradual} release method: (i) release an intermediate result for $\epsilon_m$ using a corresponding base mechanism; (ii) for $j \in [1, m-1]$, further release an intermediate result under budget $\epsilon_j - \epsilon_{j+1}$ using the base mechanism; (iii) a weighted sum of the first $m-j+1$ intermediate results is used as the query result for $\epsilon_j$. Specifically, each intermediate result is weighted inversely to its variance, so as to minimize the variance of the weighted-sum result \cite{yiwen2018utility}.

\textbf{Privacy list.} We consider both grid and random privacy levels. The \emph{grid} cases assume that the privacy list $L$ is uniformly separated by $\alpha$ (denoted as $\text{grid}(\alpha)$), i.e., $L = \{m\alpha, (m-1)\alpha, \dots, \alpha\}$. We vary $\alpha$ from $0.1$ to $0.8$ and $m$ from $1$ to $20$.  The \emph{random} cases include budgets uniformly sampled from $[0.01, C]$ (denoted as $\text{uniform}(C)$) and Gaussian-sampled random budgets centered at $\mu$ with variance $\rho^2$ (denoted as $\text{normal}(\mu, \rho^2)$), where values below $0.01$ are omitted. These settings cover many scenarios where privacy budgets vary across high and low regimes.

\textbf{Metrics.} We evaluate the utility of query results by the mean squared error $\mathbb{E}[\|r_i - q(X)\|_2^2]$. Specifically, we use both the total mean squared error (\emph{Total MSE}) defined as $\sum_{i \in [m]} \mathbb{E}[\|r_i - q(X)\|_2^2]$ and the maximum MSE ratio among all tiers (\emph{Max MSE Ratio}):
\[
\max_{i \in [m]} \frac{\mathbb{E}[\|r_i - q(X)\|_2^2]}{\mathbb{E}[\|\mathcal{M}^*(X, \epsilon_i) - q(X)\|_2^2]},
\]
where $\mathcal{M}^*$ denotes the optimal  single-tier mechanism (e.g., the base two-sided Geometric and Subset mechanisms for $\epsilon_i$ in our context). The total MSE is often dominated by the error of low-budget results, and the maximum MSE ratio cares more about individual excess error compared to solo query.

\subsection{Multi-tier DP Count Queries with Noise-adding Mechanisms}

Consider the most fundamental count query with sensitivity $\Delta = 1$, for which both the Laplace and two-sided Geometric mechanisms are applicable. We compare their multi-tier versions with the gradual baseline. In Figures \ref{fig:count} and \ref{fig:countrandom}, we present the MSE results for tier $1$ under grid privacy lists and the maximum MSE ratio results under random privacy lists, respectively. Our multi-tier two-sided Geometric mechanism outperforms the gradual approach by a large margin, as the gradual approach splits the budget into multiple parts and introduces more noise. As expected, the multi-tier two-sided Geometric mechanism also achieves significantly lower MSE than the multi-tier Laplace mechanism, especially when the privacy budgets are relatively large.

\begin{figure}[t]
\centering
\includegraphics[width=0.99\linewidth]{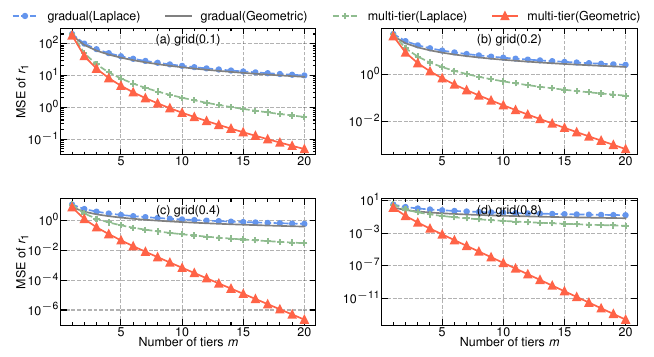}
\vspace*{-1.1em}
\caption{MSE result of count queries under grid budgets.}
\label{fig:count}
\end{figure}

\begin{figure}[t]
\centering
\includegraphics[width=0.99\linewidth]{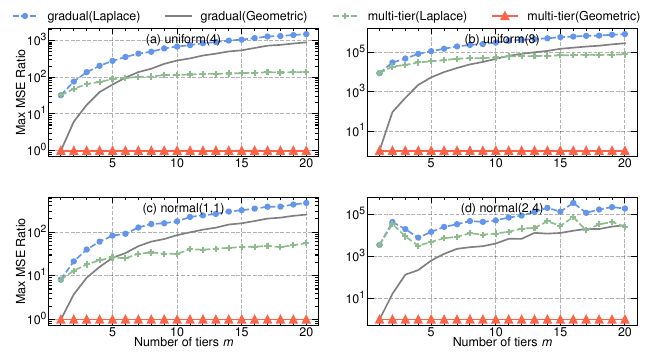}
\vspace*{-1.1em}
\caption{Maximum MSE result of count queries under random budgets (averaged over $500$ simulations).}
\label{fig:countrandom}
\end{figure}





\subsection{Multi-tier LDP Count Queries with the Subset Mechanism}

Considering the release of categorical data $x \in [d]$ using the subset mechanism, we compare the proposed multi-tier subset mechanism with the gradual baseline (using the subset mechanism as the base mechanism) and an ideal mechanism (denoted as \emph{ideal}): each $r_i$ for level $\epsilon_i$ is released by the optimal subset mechanism $\text{Subset}(x, k^*_i, \epsilon_i)$, where $k_i^*$ is the optimal subset size for $\epsilon_i$. Note that the \emph{ideal} mechanism violates the collusion-resistance property; we use it as an MSE lower bound to demonstrate the excess utility loss due to coordination in our multi-tier proposal. We also compare it with the  multi-tier binary randomized response \cite{duchi2018minimax}  (denoted as \emph{BRR}) and generalized randomized response \cite{kairouz2016discrete} (denoted as \emph{GRR}) mechanisms, using the re-sampling technique presented in \cite{xiao2009optimal}. We denote our proposal using the template generation procedure in Algorithm \ref{alg:templategen} as $\textit{multi-tier, Subset}^L_1$, and the one using Algorithm \ref{alg:templategen2} as $\textit{multi-tier, Subset}^L_2$.

\subsubsection{Total MSE results. } Under grid privacy lists, we vary the domain size and present the total MSE results in Figures \ref{fig:subset15} and \ref{fig:subset128}, respectively. In most cases, the $\textit{multi-tier, Subset}^L_1$ and $\textit{multi-tier, Subset}^L_2$  has total MSEs close to the ideal one. Since total MSE is often dominated by low-budget results, and the base BRR is comparable to the Subset mechanism in the low-budget regimes, the multi-tier BRR only has small gap to our proposals; however, when the lowest budget is relatively large (e.g., with $grid(0.8)$ privacy list), the error gap gets larger. As contrast, since GRR performs poorly in the low-budget regime, their total MSE has significant gap with our proposals in all settings. The $gradual$ approach is comparable to our proposals when $m$ is small, but the gap drastically grows as the number of tiers gets large.

\begin{figure}[t]
\centering
\includegraphics[width=0.98\linewidth]{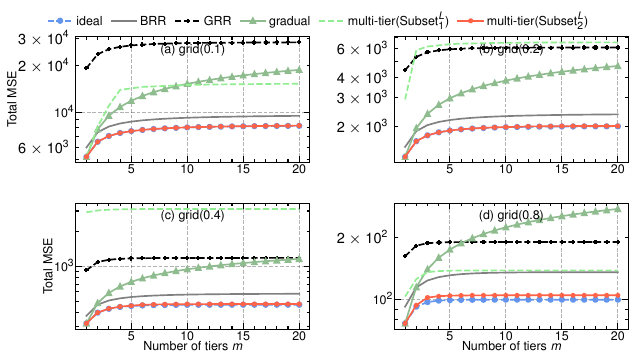}
\vspace*{-1.6em}
\caption{Total MSE result of LDP categorical queries with Subset mechanism and $d=15$, under random privacy lists.}
\label{fig:subset15}
\end{figure}

\begin{figure}[t]
\centering
\includegraphics[width=0.98\linewidth]{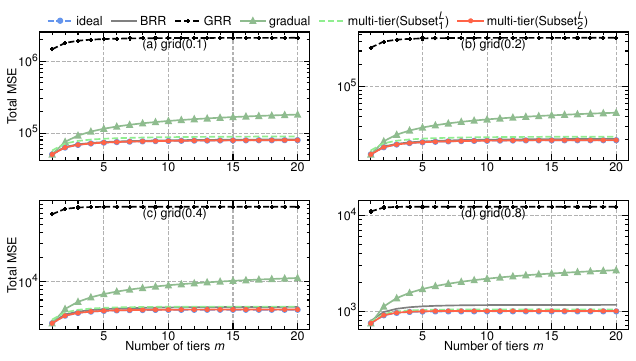}
\vspace*{-1.6em}
\caption{Total MSE result of LDP categorical queries with Subset mechanism and $d=128$, under grid privacy lists.}
\label{fig:subset128}
\end{figure}

\begin{figure}[t]
\centering
\includegraphics[width=0.98\linewidth]{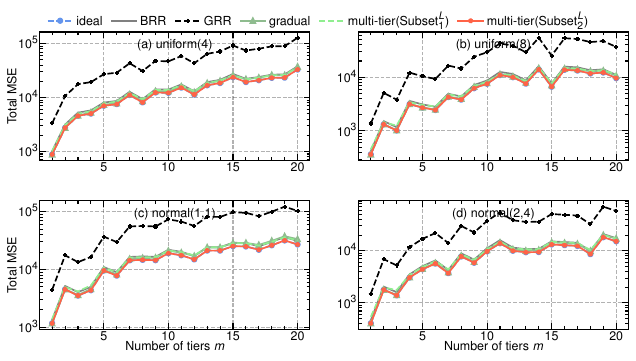}
\vspace*{-1.6em}
\caption{Total MSE result of LDP categorical queries with Subset mechanism and $d=15$, under random privacy lists.}
\label{fig:subset15random}
\end{figure}



\begin{figure}[t]
\centering
\includegraphics[width=0.98\linewidth]{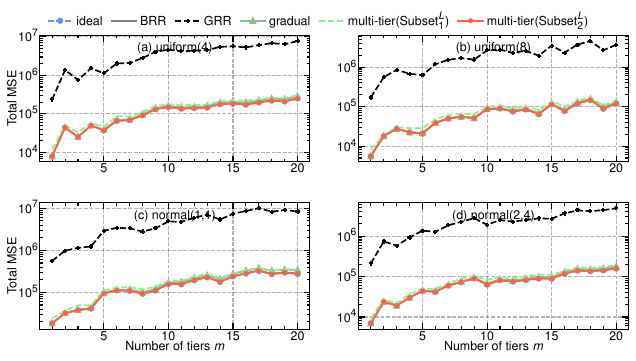}
\vspace*{-1.6em}
\caption{Total MSE result of LDP categorical queries with Subset mechanism and $d=128$, under random privacy lists.}
\label{fig:subset128random}
\end{figure}

In random privacy lists, we vary the domain size and present the total MSE results in Figures \ref{fig:subset15random} and \ref{fig:subset128random} where each result is averaged in 500 simulations. Since the error of the sampled lowest-budget result dominates total MSE, and the lowest budget can be $0.01$, both the BRR and the gradual approach is comparable to our proposals.

\subsubsection{Maximum MSE ratio results. } In grid privacy lists, we present the maximum MSE ratio results in Figures \ref{fig:subsetmaxratio15} and \ref{fig:subsetmaxratio128}, respectively. When there are relatively small number of categories (i.e., $d=15$) and the subset size parameter $k$ in the Subset has significant effects, the $\textit{multi-tier, Subset}^L_2$ has much lower maximum ratio than the $\textit{multi-tier, Subset}^L_1$ through generating more templates, and is very close to the ideal one (with the ratio controlled near to $1$). The BRR is comparable to $\textit{multi-tier, Subset}^L_2$ when all privacy budgets are relatively low (e.g., in $grid(0.1)$), but the gap are significant when the largest budget gets large. Since the base GRR has relatively large errors in the low-budget regime, it has significant gaps with the $\textit{multi-tier, Subset}^L_2$ when $\alpha\in [0.1, 0.2, 0.4]$, though the gap is relatively small in the $d=15$ \& $grid(0.8)$ setting. When there are relatively large number of categories (i.e., $d=15$), the gap between our two proposals are quite small (as neighboring subset size $k$ has small error effects in Subset), and the maximum ratios of both proposals are below $2$; the base GRR will always has large gap with our proposals, as base GRR performs much worse than Subset when $d$ is relative large even under large budget (e.g., $1<\epsilon<\log(d)$).

\begin{figure}[t]
\centering
\includegraphics[width=0.98\linewidth]{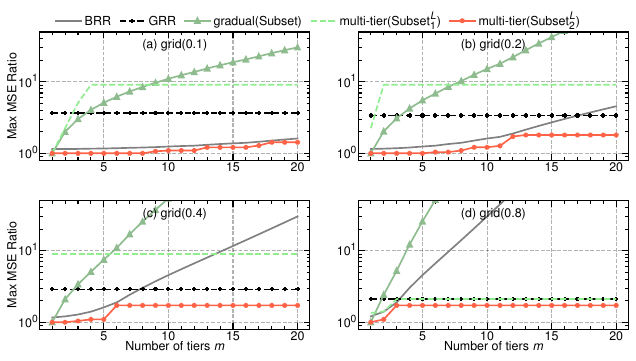}
\vspace*{-1.6em}
\caption{Maximum MSE ratio result of LDP categorical queries with Subset mechanism and $d=15$.}
\label{fig:subsetmaxratio15}
\end{figure}

\begin{figure}[t]
\centering
\includegraphics[width=0.98\linewidth]{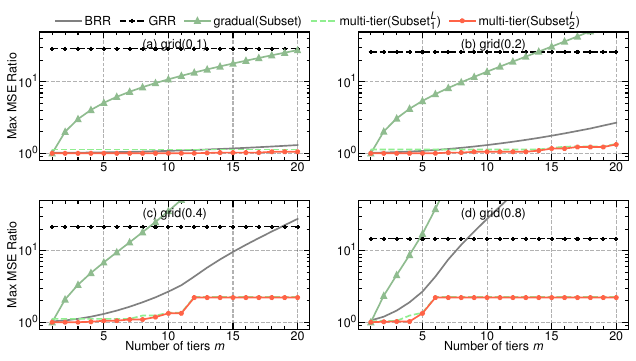}
\vspace*{-1.5em}
\caption{Maximum MSE ratio result of LDP categorical queries with Subset mechanism and $d=128$.}
\label{fig:subsetmaxratio128}
\end{figure}

In random privacy lists, we present maximum MSE ratio results in Figures \ref{fig:subsetmaxratiorandom15} and \ref{fig:subsetmaxratiorandom128}. The $\textit{multi-tier, Subset}^L_2$ is able to keep a low maximum ratio at $\leq 1.4$, and the $\textit{multi-tier, Subset}^L_1$ has relatively large gap with it when the number of categories is low (i.e., $d=15$). As expected, the gap between BRR and $\textit{multi-tier, Subset}^L_2$ is moderate when the mean budget is relatively small (e.g., in normal(1,1) and uniform(4)), and the gap grows with the mean privacy budget. As contrast, the gap with GRR and  $\textit{multi-tier, Subset}^L_2$ decreases when the mean budget gets large (e.g., in uniform(8) and normal(2,4)) and the number of categories is relatively small (i.e., $d=15$).

\begin{figure}[h]
\centering
\includegraphics[width=0.98\linewidth]{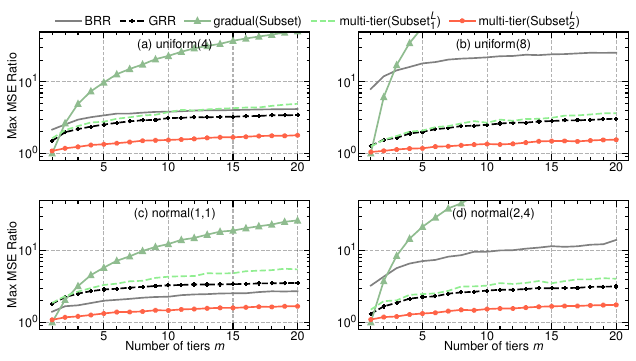}
\vspace*{-1.5em}
\caption{Maximum MSE ratio result of LDP categorical queries with Subset mechanism and $d=15$ under random privacy lists (averaged over 500 simulations).}
\label{fig:subsetmaxratiorandom15}
\end{figure}



\begin{figure}[h]
\centering
\includegraphics[width=0.98\linewidth]{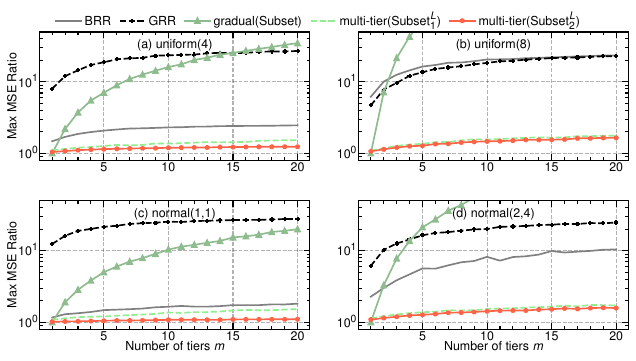}
\vspace*{-1.5em}
\caption{Maximum MSE ratio result of LDP categorical queries with Subset mechanism and $d=128$ under random privacy lists (averaged over 500 simulations).}
\label{fig:subsetmaxratiorandom128}
\end{figure}

In summary, since our multi-tier Subset design tightly approximates its optimal hyper-parameters (i.e., subset size, privacy level)  as in the single-tier Subset mechanism for all privacy levels through template generation, our proposals outperform baselines with a clear margin (both in total MSE and maximum MSE ratios).  

\section{Related Works}\label{sec:related}
\textbf{Differentially private query mechanisms. } Depending on the query types (e.g., discrete/continuous, one/multiple dimension, centralized/decentralized), plenty of base DP mechanisms are available in the literature. In the centralized curator model, for count and numerical summation queries, the Laplace mechanism \cite{dwork2006differential} and the Gaussian mechanism \cite{dwork2008differential} provide pure and approximate DP, respectively; the two-sided Geometric mechanism \cite{ghosh2012universally} further improves error rates for count query, and is universally optimal under sensitivity $\Delta=1$; the staircase mechanism \cite{geng2015staircase} provides optimal utility for numerical summation queries; for selection queries (e.g., selecting the candidate with the maximum votes), the exponential mechanism \cite{mcsherry2007mechanism} is a general tool for designing differentially private selection probabilities. In the local model of DP where each user sanitizes their own data, for count queries over categorical data, randomized response \cite{warner1965randomized} and its variants \cite{kairouz2016discrete,duchi2018minimax} are prevalent mechanisms, the local hash mechanism \cite{wang2017locally} improves communication efficiency by introducing shared randomness via hash functions, and the Subset mechanism \cite{wang2016mutual,wang2019local} provides optimal utility \cite{ye2018optimal}; for numerical queries, several mechanisms \cite{duchi2018minimax,bhowmick2018protection,wang2019collecting,li2020estimating} provide optimal utility in the high privacy regime.

\textbf{Multi-tier DP queries. } The idea of support multi-analyst DP query previously appears in database settings, such as for database queries \cite{koufogiannis2016gradual,zhang2023dprovdb}, dataset perturbation \cite{xiao2009optimal}. Specifically, the \cite{koufogiannis2016gradual} provides a multi-tier realization of the Laplace mechanism,  \cite{zhang2023dprovdb} uses multi-tier Gaussian noise adding, the \cite{xiao2009optimal} proposes re-sampling strategies for multi-tier randomized response. Nevertheless, these base DP mechanisms are known to be suboptimal. For example, in the low-privacy regime, Laplace mechanism exhibits an exponential error gap compared to the two-sided Geometric mechanism \cite{ghosh2012universally}; meanwhile, the binary randomized response \cite{warner1965randomized} and general randomized response \cite{kairouz2016discrete}  perform significantly worse than the Subset mechanism \cite{wang2019local} in the low-privacy and high-privacy regimes, respectively. In multi-analyst scenarios, the target privacy levels $\{\epsilon_i\}_{i\in [m]}$ can span from low regimes (e.g., $\epsilon > 1$) to high regimes (e.g., $\epsilon < 1$), making these approaches inherently incapable of achieving the utility objective. 

A closely-related setting is \emph{accuracy-first} DP query \cite{ligett2017accuracy,whitehouse2022brownian,ghazi2025private}, which aims to minimize privacy losses upon satisfying the analyst's accuracy criteria. It firstly releases a high-privacy result, then trace back to low-privacy result according to the accuracy requirements. The \cite{ligett2017accuracy,whitehouse2022brownian} propose accuracy-first Laplace and Gaussian mechanisms, respectively, based on generating a continuous trace of the multi-tier transformation (e.g., in \cite{koufogiannis2016gradual}) from the lowest-privacy result to higher privacy results. Our proposals naturally support accuracy-first settings with (non-continuous) privacy levels, by calling our multi-tier mechanisms with a fine-grained privacy list. It is interesting to extend our results for two-sided Geometric mechanism in Section \ref{sec:noise} to the continuous accuracy-first setting.

\section{Conclusion}\label{sec:conclusion}
In this paper, we investigated the multi-tier DP query release problem, which is pervasive in data warehouses, mobile services, and data markets. For noise-adding DP mechanisms, we developed a general transformation framework via characteristic functions,  provided concrete instantiations for count queries with two-sided Geometric noises, and demonstrated infeasible for some prevalent mechanisms (e.g., discrete Gaussian). Additionally, for the more challenging mapping-based DP mechanisms (commonly found in local DP), we provided a template-based multi-tier method and achieved optimal error rates for count queries with the Subset mechanism in local DP. Through extensive experiments, we numerically evaluated the effectiveness of our multi-tier DP proposals.



%% file: appendix.tex
\appendix




\section{Proof of Theorem \ref{the:noiseaddingprivacy}}\label{app:noiseaddingprivacy}
if the algorithm fails at any step, it outputs nothing and the failure event is independent of the true query answer $q(X)$ or the dataset $D$, thus the total privacy loss is $0$.

When the algorithm succeed, for any neighboring datasets $X,X'$ and any $S\subseteq [m]$, let  $\mathcal{M}_S$ denote the subset $\{r_i\}_{i\in S}$ of the algorithm's output $\{r_i\}_{i\in [m]}$ with input $q(X)$, and let $\mathcal{M}'_S$ denote the subset $\{r_i\}_{i\in S}$ of the algorithm's output $\{r_i\}_{i\in [m]}$ with input $q(X')$. Let $i^* = \min S$ (the smallest index in $S$, we define a post-processing function $f:\mathcal{Y}\mapsto \mathcal{Y}^{|S|}$ as follows:
\begin{enumerate}
\item setting $r_{i^*}$ as the input;
\item sequentially, for $i'\in [i^*+1, \max S]$, drawing an  residual noise $s_{i'-1}$ according to Equation \ref{eq:tsgresidual}, and set $r_{i'} = r_{i'-1} + s_{i'-1}$;
\item return results $\{r_i\}_{i\in [i^*,\max S]\backslash S}$ that the indices are in $S$.
\end{enumerate}
The residual noises $s_{i'-1}$ are predetermined and do not depend on the dataset. Therefore $f$ is a fixed Markov kernel that does not use any additional information about $X$ or $X'$.

By construction, for any dataset $X$, we have
\[
\mathcal{M}_S(X,L) \overset{\mathrm{d}}{=} f(\mathcal{M}_{\{i^*\}}(X,L)),
\]
because the mechanism generates $r_{i^*}$ first and then produces all later $r_i$ via the same additive residual noises. The same holds for $X'$. Applying the data processing inequality to the post-processing function $f$ yields
\[
D\bigl(\mathcal{M}_S(X,L)\,\bigr\|\,\mathcal{M}_S(X',L)\bigr)
\;\le\; D\bigl(\mathcal{M}_{\{i^*\}}(X,L)\,\bigr\|\,\mathcal{M}_{\{i^*\}}(X',L)\bigr).
\]
Thus the desired inequality holds.

\section{Proof of Lemma \ref{lemma:expansion}}\label{app:expansion}
For any subset $S' \subseteq [d]$ of size $k+1$, the probability that $r' = S'$ can be computed by considering which element was added during expansion. There are exactly $k+1$ possible ways that $S'$ could have been generated. Specifically, for each $y \in S'$, there is exactly one subset $S' \setminus \{y\}$ of size $k$ that could have been expanded to $S'$ by adding $y$. Thus, we have:
\[
\Pr[r' = S' | x] = \sum_{y \in S'} \Pr[r = S' \setminus \{y\} | x] \cdot \frac{1}{d-k}.
\]
Now consider two cases:

\textbf{Case 1 ($x \in S'$): }
There is exactly one element $y = x$ in $S'$ such that $x \in S' \setminus \{y\}$ (when $y \neq x$), and for $y = x$, we have $x \notin S' \setminus \{x\}$. Using the probability distribution of $Subset(x,k,\epsilon)$:
\begin{align*}
\Pr[r'\!=\!S'|x]&= \frac{1}{d\!-\!k} \left[ \Pr[r\!=\!S'\!\setminus\!\{x\}|x]\!+\!\sum_{y\in S', y\neq x}\Pr[r\!=\!S'\!\setminus\!\{y\}|x] \right] \\
&= \frac{1}{d-k} \left[ \frac{1}{\binom{d-1}{k-1}e^\epsilon + \binom{d-1}{k}} + k \cdot \frac{e^\epsilon}{\binom{d-1}{k-1}e^\epsilon + \binom{d-1}{k}} \right] \\
&= \frac{1}{d-k} \cdot \frac{1 + k e^\epsilon}{\binom{d-1}{k-1}e^\epsilon + \binom{d-1}{k}}\\
&=\frac{(1+k e^\epsilon)/(k+1)}{\binom{d-1}{k}e^\epsilon + \binom{d-1}{k+1}}.
\end{align*}

\textbf{Case 2 ($x \notin S'$): }
For all $y \in S'$, we have $x \notin S' \setminus \{y\}$, so:
\begin{align*}
\Pr[r' = S' | x] &= \frac{1}{d-k} \sum_{y \in S'} \Pr[r = S' \setminus \{y\} | x] \\
&= \frac{1}{d-k} \cdot \frac{k+1}{\binom{d-1}{k-1}e^\epsilon + \binom{d-1}{k}}\\
&=\frac{1}{\binom{d-1}{k}e^\epsilon + \binom{d-1}{k+1}}.
\end{align*}

Now, for the target distribution $Subset(x,k+1,\epsilon')$ with $\epsilon' = \ln(\frac{e^{\epsilon}k+1}{k+1})$, we have:
\[
\Pr[S' | x] = \begin{cases}
    \frac{e^{\epsilon'}}{\binom{d-1}{k}e^{\epsilon'} + \binom{d-1}{k+1}} & \text{if } x \in S' \\
    \frac{1}{\binom{d-1}{k}e^{\epsilon'} + \binom{d-1}{k+1}} & \text{if } x \notin S'
\end{cases}
\]
Therefore, $r' \overset{\mathrm{d}}{=} Subset(x,k+1,\ln((1+k e^\epsilon)/(k+1)))$.

\section{Proof of Lemma \ref{lemma:rescale}}\label{app:rescale}
Let $r' = \text{Rescale}(r,\beta)$. For any subset $S \subseteq [d]$ of size $k$, the probability that $r' = S$ can be computed by considering the two cases in the rescale operation:
\[
\Pr[r' = S | x] = \beta \cdot \Pr[r = S | x] + (1-\beta) \cdot \frac{1}{\binom{d}{k}}.
\]
Now consider two cases:

\textbf{Case 1 ($x \in S$): }
Using the probability distribution of $Subset(x,k,\epsilon)$:
\begin{align*}
\Pr[r' = S | x] &= \beta \cdot \frac{e^{\epsilon}}{\binom{d-1}{k-1}e^{\epsilon} + \binom{d-1}{k}} + (1-\beta) \cdot \frac{1}{\binom{d}{k}} \\
&= \frac{\beta e^{\epsilon}}{\binom{d-1}{k-1}e^{\epsilon} + \binom{d-1}{k}} + \frac{1-\beta}{\binom{d}{k}}
\end{align*}

\textbf{Case 2 ($x \notin S$): }
\begin{align*}
\Pr[r' = S | x] &= \beta \cdot \frac{1}{\binom{d-1}{k-1}e^{\epsilon} + \binom{d-1}{k}} + (1-\beta) \cdot \frac{1}{\binom{d}{k}} \\
&= \frac{\beta}{\binom{d-1}{k-1}e^{\epsilon} + \binom{d-1}{k}} + \frac{1-\beta}{\binom{d}{k}}
\end{align*}

The new probability ratio is then:
\begin{align*}
\frac{\frac{\beta e^{\epsilon}}{N} + \frac{1-\beta}{M}}{\frac{\beta}{N} + \frac{1-\beta}{M}}
= \frac{\beta e^{\epsilon} + \frac{N}{M}(1-\beta)}{\beta + \frac{N}{M}(1-\beta)}
\end{align*}
where $N = \binom{d-1}{k-1}e^{\epsilon} + \binom{d-1}{k}$ and $M = \binom{d}{k}$. Using the identity $\frac{N}{M} = \frac{\binom{d-1}{k-1}(e^{\epsilon} + \frac{d-k}{k})}{\binom{d}{k}} = \frac{k}{d}(e^{\epsilon} + \frac{d-k}{k}) = \frac{k e^{\epsilon} + d - k}{d}$, we get the new privacy level $\epsilon'$ as:
\[
e^{\epsilon'} = \frac{d\beta e^{\epsilon} + (1-\beta)(k e^{\epsilon} + d - k)}{d\beta + (1-\beta)(k e^{\epsilon} + d - k)}
\]

\section{Proof of Theorem \ref{the:utilitymatch}}\label{app:utilitymatch}
Let $k^*_i$ denote the optimal subset size for the  $\epsilon_i$, we consider $4$ cases about the privacy budget $\epsilon_i$.

\textbf{Case 1: } In the case of $\epsilon_i> \ln(d-1)$, we have $k_i^*\equiv 1$. It is handled in the step 1 with only rescale operations, then the corresponding result $s$ is distributionally equivalent to $Subset(\epsilon_i,1)$. Therefore, we have $\mathbb{E}[\|Subset^{L}_i-x\|_2^2]= V(\epsilon_i,k^*_i)$.

\textbf{Case 2: } In the case of $\ln(d-1)\geq \epsilon_i\geq \ln(d/\lfloor (d-1)/2\rfloor)$, we have $\lceil d/(e^{\epsilon_i}+1)\rceil\geq k^*_i\geq \lfloor d/(e^{\epsilon_i}+1)\rfloor$. It is handled in the step 2 by the generated templates in Algorithm \ref{alg:templategen}, the corresponding result $s$ has either size $\lfloor d/(e^{\epsilon_i}+1)\rfloor$ (with tight privacy budget $\ln(d/\lfloor d/(e^{\epsilon_i}+1)\rfloor-1)$) or size $\lceil d/(e^{\epsilon_i}+1)\rceil$ (with tight privacy budget $\ln(d/\lceil d/(e^{\epsilon_i}+1)\rceil-1)$). Since the following two inequality holds:
$$V(\ln(d/\lceil d/(e^{\epsilon_i}+1)\rceil-1),\lceil d/(e^{\epsilon_i}+1)\rceil)\geq V(\epsilon_i,k^*_i);$$
$$V(\epsilon_i,k^*_i)\geq V(\ln(d/\lfloor d/(e^{\epsilon_i}+1)\rfloor-1), \lfloor d/(e^{\epsilon_i}+1)\rfloor).$$
Now consider a formula  $\frac{V(\ln(d/(k+1)-1), k+1)}{V(\ln(d/k-1), k)}$, we have:
\begin{alignat*}{2}
&\frac{V(\ln(d/(k+1)-1), k+1)}{V(\ln(d/k-1), k)}\\
=&\frac{-4 (1 + k)^2 + d (3 + 4 k)}{(d - 
   2 (1 + k))^2}\Big/\frac{-4 k^2 + d (-1 + 4 k)}{(d - 2 k)^2 }.
\end{alignat*}
Let $F(k)=\ln(\frac{-4 k^2 + d (-1 + 4 k)}{(d - 2 k)^2})$, for $k\in [1,\lfloor (d-1)/2\rfloor]$, we have $\frac{\partial F}{\partial k}\geq 0$, $\frac{\partial^2 F}{\partial k^2}\leq 0$ when $k\leq \frac{3d-\sqrt{3(d^2-d)}}{6}$, and $\frac{\partial^2 F}{\partial k^2}\geq 0$ when $k\geq \frac{3d-\sqrt{3(d^2-d)}}{6}$. Therefore, when $\lfloor (d-1)/2\rfloor\geq 2$ (i.e., $d\geq 5$), we have:
\begin{alignat*}{2}
&\frac{V(\ln(d/(k+1)-1), k+1)}{V(\ln(d/k-1), k)}\\
\leq &\exp(F(k+1)-F(k))\\
\leq &\max\{\exp(F(2)-F(1)), \exp(F(\lfloor (d-1)/2\rfloor)-F(\lfloor (d-1)/2\rfloor-1))\}\\
\leq &\frac{(7d-16)(d-2)^2}{(3d-4)(d-4)^2}.
\end{alignat*}
Consequently, we have:
\begin{alignat*}{2}
&\mathbb{E}[\|Subset^{L}_i-x\|_2^2]\\
\leq& V(\ln(d/\lceil d/(e^{\epsilon_i}+1)\rceil-1), \lceil d/(e^{\epsilon_i}+1)\rceil)\\
\leq&\frac{V(\ln(d/\lceil d/(e^{\epsilon_i}+1)\rceil-1), \lceil d/(e^{\epsilon_i}+1)\rceil)}{V(\ln(d/\lfloor d/(e^{\epsilon_i}+1)\rfloor-1), \lfloor d/(e^{\epsilon_i}+1)\rfloor)}\cdot V(\epsilon_i,k^*_i)\\
\leq& \frac{(7d-16)(d-2)^2}{(3d-4)(d-4)^2}\cdot V(\epsilon_i,k^*_i)\\
\leq& \frac{(7\cdot 5-16)(5-2)^2}{(3\cdot 5-4)(5-4)^2}\cdot V(\epsilon_i,k^*_i)\\
\leq& \frac{171}{11}\cdot V(\epsilon_i,k^*_i)
\end{alignat*}

\textbf{Case 3: } In the case of $d$ is even and $\ln(d/\lfloor (d-1)/2\rfloor-1)> \epsilon_i\geq \ln((d+2)/\lfloor d/2\rfloor-1)$. Note that this case only applies to when $d$ is even and $d\geq 4$. It is handled at the start of Step 3 (see lines 8-10 in Algorithm \ref{alg:multitiersubset}), and the corresponding result will matches to the resulting $s$ at line 10. Therefore, the following two inequalities hold:
$$V(\ln((d+2)/\lfloor d/2\rfloor-1),\lfloor d/2\rfloor)\geq V(\epsilon_i,k^*_i);$$
$$V(\epsilon_i,k^*_i)\geq V(\ln(d/\lfloor (d-1)/2\rfloor-1),\lfloor (d-1)/2\rfloor).$$
Notice that when $d$ is even and $d\geq 4$ (the subcase $d=2$ is trivial as the optimal subset size is always $1$), we have:
\begin{alignat*}{2}
&\frac{V(\ln((d+2)/\lfloor d/2\rfloor-1),\lfloor d/2\rfloor)}{V(\ln(d/\lfloor (d-1)/2\rfloor-1),\lfloor (d-1)/2\rfloor)}\\
=&\frac{V(\ln(2(d+2)/d-1), d/2)}{V(\ln(2d/(d-2)-1),(d-2)/2)}\\
=&1+\frac{2}{d}+\frac{2(3+d)}{d(d-1)-4}\\
\leq& 1+\frac{2}{4}+\frac{2(3+4)}{4(4-1)-4}\\
\leq& 13/4.
\end{alignat*}
Consequently, we have:
\begin{alignat*}{2}
&\mathbb{E}[\|Subset^{L}_i-x\|_2^2]\\
\leq& V(\ln((d+2)/\lfloor d/2\rfloor-1),\lfloor d/2\rfloor)\\
\leq& \frac{V(\ln((d+2)/\lfloor d/2\rfloor-1),\lfloor d/2\rfloor)}{V(\ln(d/\lfloor (d-1)/2\rfloor-1),\lfloor (d-1)/2\rfloor)}\cdot V(\epsilon_i,k^*_i)\\
\leq& (1+\frac{2}{d}+\frac{2(3+d)}{d(d-1)})\cdot V(\epsilon_i,k^*_i)\\
\leq& \frac{13}{4}\cdot V(\epsilon_i,k^*_i).
\end{alignat*}

\textbf{Case 4: } In the case of $\epsilon_i< \ln((d+2)/\lfloor d/2\rfloor-1)$ (if $d$ is even) or $e^{\epsilon_{i}}+1< d/\lfloor (d-1)/2\rfloor$ (if $d$ is odd), we have $k^*_i=\lfloor d/2\rfloor$ (if $d$ is even) or $k^*_i=\lfloor (d-1)/2\rfloor$ (if $d$ is odd). It is handled in the Step 4, with only rescale operations, then the corresponding result $s$ is distributionally equivalent to $Subset(\epsilon_i,k^*_i)$. Therefore, we have $\mathbb{E}[\|Subset^{L}_i-x\|_2^2]= V(\epsilon_i,k^*_i)$.

In summary, we have the conclusion.

\section{Proof of Theorem \ref{the:utilitymatch2}}\label{app:utilitymatch2}
For any $k,k+1\in [d]$, we have:
\begin{alignat*}{2}
&V(\epsilon,k)-V(\epsilon,k+1)\\
=&\frac{(d-1)^2 \left(d^2-d (2 k+1)-\left(e^{2 \epsilon }-1\right) k (k+1)\right)}{\left(e^{\epsilon }-1\right)^2 k (k+1) (d-k-1) (d-k)},
\end{alignat*}
which implies that the optimal subset size $k^*$ is non-increasing with $\epsilon$ and the equality holds when:
\begin{equation}\label{}
\epsilon_{(k)}=\log\left(\sqrt{\frac{d-k}{k}\cdot \frac{d-k-1}{k+1}}\right)
\end{equation}

Let $k^*_i$ denote the optimal subset size for the  $\epsilon_i$, we now consider $4$ cases about the privacy budget $\epsilon_i$.

\textbf{Case 1: } In the case of $\epsilon_i> \ln(d-1)$, we have $k_i^*\equiv 1$. It is handled in the step 1 with only rescale operations, then the corresponding result $s$ is distributionally equivalent to $Subset(\epsilon_i,1)$. Therefore, we have $\mathbb{E}[\|Subset_1^{L}(x)_i-x\|_2^2]= V(\epsilon_i,k^*_i)$.

\textbf{Case 2: } In the case of $\ln(d-1)\geq \epsilon_i\geq \ln(d/\lfloor (d-1)/2\rfloor)$, we have $\lceil d/(e^{\epsilon_i}+1)\rceil\geq k^*_i\geq \lfloor d/(e^{\epsilon_i}+1)\rfloor$, and there are three subcases: 

(I) In the first subcase that $\lceil d/(e^{\epsilon_i}+1)\rceil=\lfloor d/(e^{\epsilon_i}+1)\rfloor$, both values equal to $k^*$ and the matched result coresponds to the template with budget $\ln(d/k^*-1)$, thus $\mathbb{E}[\|Subset_2^{L}(x)_i-x\|_2^2]= V(\epsilon_i,k^*_i)$.  

(II) In the second subcase that $\lceil d/(e^{\epsilon_i}+1)\rceil>\lfloor d/(e^{\epsilon_i}+1)\rfloor$ and $\ln((d+1)/\lceil d/(e^{\epsilon_i}+1)\rceil-1)\geq \epsilon_i\geq \ln(d/\lceil d/(e^{\epsilon_i}+1)\rceil-1)$, it is handled by the lines 9-13. Since $\epsilon_{(\lfloor d/(e^{\epsilon_i}+1)\rfloor)}\geq \ln((d+1)/\lceil d/(e^{\epsilon_i}+1)\rceil-1)$, we have $V(\epsilon_i, \lceil d/(e^{\epsilon_i}+1)\rceil)\leq V(\epsilon_i, \lfloor d/(e^{\epsilon_i}+1)\rfloor)$, the $k^*=\lceil d/(e^{\epsilon_i}+1)\rceil$ is an optimal subset size for $\epsilon_i$. Consequently, the matched result for $\epsilon_i$ at the lines 9-13 is distributionally equivalent to $Subset(d, \lceil d/(e^{\epsilon_i}+1)\rceil, \epsilon_i)$, and thus $\mathbb{E}[\|Subset_2^{L}(x)_i-x\|_2^2]= V(\epsilon_i,k^*_i)$.

(III) In the third subcase that $\lceil d/(e^{\epsilon_i}+1)\rceil>\lfloor d/(e^{\epsilon_i}+1)\rfloor$ and $\ln(d/\lfloor d/(e^{\epsilon_i}+1)\rfloor-1)\geq \epsilon_i\geq \ln((d-1)/\lfloor d/(e^{\epsilon_i}+1)\rfloor-1)$, it is also handled by the lines 9-13. Since $\ln((d-1)/\lfloor d/(e^{\epsilon_i}+1)\rfloor-1)\geq \epsilon_{(\lfloor d/(e^{\epsilon_i}+1)\rfloor)}$, we have $V(\epsilon_i, \lceil d/(e^{\epsilon_i}+1)\rceil)\geq V(\epsilon_i, \lfloor d/(e^{\epsilon_i}+1)\rfloor)$, the $k^*=\lfloor d/(e^{\epsilon_i}+1)\rfloor$ is an optimal subset size for $\epsilon_i$. Consequently, the matched result for $\epsilon_i$ at the lines 9-13 is distributionally equivalent to $Subset(d, \lfloor d/(e^{\epsilon_i}+1)\rfloor, \epsilon_i)$, and thus $\mathbb{E}[\|Subset_2^{L}(x)_i-x\|_2^2]= V(\epsilon_i,k^*_i)$. 

(IV) In the fourth subcase that $\lceil d/(e^{\epsilon_i}+1)\rceil>\lfloor d/(e^{\epsilon_i}+1)\rfloor$ and $\ln((d-1)/\lfloor d/(e^{\epsilon_i}+1)\rfloor-1) > \epsilon_i>\ln((d+1)/\lceil d/(e^{\epsilon_i}+1)\rceil-1)$. It will be matched to a result with budget $\ln((d+1)/\lceil d/(e^{\epsilon_i}+1)\rceil-1)$ and subset size $\lceil d/(e^{\epsilon_i}+1)\rceil$. Since following two inequalities hold:
$$V(\ln((d+1)/\lceil d/(e^{\epsilon_i}+1)\rceil-1),\lceil d/(e^{\epsilon_i}+1)\rceil)\geq V(\epsilon_i,k^*_i);$$
$$V(\epsilon_i,k^*_i)\geq V(\ln((d-1)/\lfloor d/(e^{\epsilon_i}+1)\rfloor-1), \lfloor d/(e^{\epsilon_i}+1)\rfloor),$$
we have:
\begin{alignat*}{2}
&\frac{\mathbb{E}[\|Subset_2^{L}(x)_i-x\|_2^2]}{ V(\epsilon_i,k^*_i)}\\
\leq&\frac{V(\ln((d+1)/(k+1)-1), k+1)}{V(\ln((d-1)/k-1), k)}\\
=&\frac{(d - k) (2 + d^2 (3 + 4 k) + k (7 + 4 k (2 + k)) - 
   d (5 + k (11 + 8 k)))}{(1 - d + k)^2 (1 - 4 k^2 + d (-1 + 4 k))}.
\end{alignat*}
In the theorem, we use the formula with $k=\lceil d/(e^\epsilon_i+1)\rceil-1$.


\textbf{Case 3: } In the case of $d$ is even and $\ln(d/\lfloor (d-1)/2\rfloor-1)> \epsilon_i\geq \ln((d+2)/\lfloor d/2\rfloor-1)$. Note that this case only applies to when $d$ is even and $d\geq 4$. It is handled at the start of Step 3 (see lines 7-9 in Algorithm \ref{alg:multitiersubset}), and  will match to the resulting $s$ at line 9. Therefore, the following two inequalities hold:
$$V(\ln((d+2)/\lfloor d/2\rfloor-1),\lfloor d/2\rfloor)\geq V(\epsilon_i,k^*_i);$$
$$V(\epsilon_i,k^*_i)\geq V(\ln(d/\lfloor (d-1)/2\rfloor-1),\lfloor (d-1)/2\rfloor).$$
Notice that when $d$ is even and $d\geq 4$ (the subcase $d=2$ is trivial as the optimal subset size is always $1$), we have:
\begin{alignat*}{2}
&\frac{V(\ln((d+2)/\lfloor d/2\rfloor-1),\lfloor d/2\rfloor)}{V(\ln(d/\lfloor (d-1)/2\rfloor-1),\lfloor (d-1)/2\rfloor)}\\
=&\frac{V(\ln(2(d+2)/d-1), d/2)}{V(\ln(2d/(d-2)-1),(d-2)/2)}\\
=&1+\frac{2}{d}+\frac{2(3+d)}{d(d-1)-4}\\
\leq& 1+\frac{2}{4}+\frac{2(3+4)}{4(4-1)-4}\\
\leq& 13/4.
\end{alignat*}
Consequently, we have:
\begin{alignat*}{2}
&\mathbb{E}[\|Subset^{L}_i-x\|_2^2]\\
\leq& V(\ln((d+2)/\lfloor d/2\rfloor-1),\lfloor d/2\rfloor)\\
\leq& \frac{V(\ln((d+2)/\lfloor d/2\rfloor-1),\lfloor d/2\rfloor)}{V(\ln(d/\lfloor (d-1)/2\rfloor-1),\lfloor (d-1)/2\rfloor)}\cdot V(\epsilon_i,k^*_i)\\
\leq& (1+\frac{2}{d}+\frac{2(3+d)}{d(d-1)})\cdot V(\epsilon_i,k^*_i)\\
\leq& \frac{13}{4}\cdot V(\epsilon_i,k^*_i).
\end{alignat*}

\textbf{Case 4: } In the case of $\epsilon_i< \ln((d+2)/\lfloor d/2\rfloor-1)$ (if $d$ is even) or $e^{\epsilon_{i}}+1< d/\lfloor (d-1)/2\rfloor$ (if $d$ is odd), we have $k^*_i=\lfloor d/2\rfloor$ (if $d$ is even) or $k^*_i=\lfloor (d-1)/2\rfloor$ (if $d$ is odd). It is handled in the Step 4, with only rescale operations, then the corresponding result $s$ is distributionally equivalent to $Subset(\epsilon_i,k^*_i)$. Therefore, we have $\mathbb{E}[\|Subset^{L}_i-x\|_2^2]= V(\epsilon_i,k^*_i)$.

In summary, we have the conclusion.

\section{Non-existence of Multi-tier Staircase Mechanism}\label{app:staircase}
In the staircase mechanism \cite{geng2015staircase} for one-dimensional numerical queries with sensitivity $\Delta$, the noise distribution is as follows  :
\begin{equation}
\text{Pr}[\text{SC}(\epsilon)\!=\!y] =\! 
\begin{cases}
a(\gamma), & 0 \le y < \gamma \Delta,\\[4pt]
a(\gamma) e^{-\epsilon}, & \gamma \Delta \le y < \Delta,\\[4pt]
e^{-k\epsilon} \text{Pr}[\text{SC}(\epsilon)\!=\!y\!-\!k\Delta], & y\!\in\! [k\Delta, (k\!+\!1)\Delta), k\!\in\!\mathbb{Z}_{\ge 0},\\[4pt]
\text{Pr}[\text{SC}(\epsilon)=-y], & y < 0,
\end{cases}
\end{equation}
where the normalization factor \begin{equation}
a(\gamma) = \frac{1 - e^{-\epsilon}}
{2\Delta \left( \gamma + e^{-\epsilon}(1-\gamma) \right)},
\end{equation}
 and the hyperparameter $\gamma=1/(e^{\epsilon/2}+1)$. Its characteristic function is:
\[
\Phi_{\epsilon}(t) := \frac{e^{-{3 \epsilon}{2}} \left(e^{\epsilon}-1\right)^2 \left(e^{\epsilon} \sin \left(\frac{\Delta t}{e^{\epsilon/2}+1}\right)+\sin \left(\Delta \left(1-\frac{1}{e^{\epsilon/2}+1}\right) t\right)\right)}{2 \Delta t (\cosh (\epsilon)-\cos (\Delta t))},
\]
and the residual function is $\Phi_{\epsilon_{i+1}}(t) / \Phi_{\epsilon_i}(t)$.
Recall that any valid characteristic function $R: \mathbb{R} \to \mathbb{C}$ must be \emph{positive definite} \cite{bochner2005harmonic}: for every $n \in \mathbb{N}$ and every choice of points $t_1, \dots, t_C \in \mathbb{R}$, the associated \emph{Bochner matrix}
\[
M := \bigl[ M_{a,b} \bigr]_{a,b \in [C]}, \qquad
M_{a,b} := R(t_a - t_b),
\]
must be Hermitian positive semi-definite. Equivalently, all eigenvalues of $M$ must be non-negative. We construct a counterexample for the residual function $$R(t) = \Phi_{\epsilon_{i+1}}(t) / \Phi_{\epsilon_i}(t)$$ with $\epsilon_{i} = 2.8$ and $\epsilon_{i+1} = 1$. When $t$ takes values in $\{5\pi \cdot (k-1)\}_{k \in [2]}$, the Bochner matrix $M \in \mathbb{C}^{2 \times 2}$ with entries $M_{a,b} = R(t_a - t_b)$ are as follow:
\[
\begin{bmatrix}
1 & -1.43887 \\
-1.43887 & 1 \\
\end{bmatrix}
\]
It has a smallest eigenvalue of approximately $-0.43887 < 0$.
Therefore, $M$ is not positive semi-definite, the $\Phi_{\epsilon_{i+1}}(t) / \Phi_{\epsilon_i}(t)$ is not a valid characteristic function, and no additive noise can transform from $\text{SC}(2.8))$ to $\text{SC}(1.0)$.

\revise{
\section{Multi-tier Local Hash Mechanism}\label{app:multitierLH}
In this section, we extend the template-based framework to the Local Hash (LH) mechanism \cite{wang2017locally}. In LH, the domain $[d]$ is mapped to a smaller domain $[g]$ via a hash function, where the optimal $g$ depends on the privacy budget $\epsilon$ (e.g., $g^*=\lfloor e^\epsilon+1\rfloor$). To facilitate multi-tier transitions, we restrict the hash domain size to powers of two, i.e., $g \in \{2^k \mid k \in \mathbb{Z}^+\}$.

\textbf{Hierarchical Hash Functions.} Let $K = \lfloor \log_2(e^{\epsilon_1} + 1) \rfloor$ be the initial domain exponent for the highest privacy budget $\epsilon_1$. We define a sequence of hash functions $\{H_k: [d] \to [2^k]\}_{k=1}^K$. Starting from a base universal hash function $H_K$, the subsequent functions are defined as:
\begin{equation}
    H_{k-1}(x) = \lceil H_k(x) / 2 \rceil, \quad \text{for } k = K, K-1, \dots, 2.
\end{equation}
This construction ensures that if $H_k(x) = j$, then its projection in the smaller domain $[2^{k-1}]$ is deterministically $H_{k-1}(x)$.

\textbf{Primitive Operators:} Pertained to the LH mechanism, we define two operators for adapting the hash domain size and privacy levels:
\begin{itemize}[leftmargin=1em]
    \item \emph{Merge Operator}: Given a result $r \in [2^k]$ ($k>1$), $\text{Merge}(r) := \lceil r/2 \rceil \in [2^{k-1}]$. This operator effectively merges two adjacent hash buckets.
    \item \emph{Rescale Operator}: Given $r \in [2^k]$ with budget $\epsilon$, the operator $\text{Rescale}(r, \beta)$ outputs $r$ with probability $\beta$, and with probability $1-\beta$ samples a value uniformly from $[2^k]$.
\end{itemize}

\begin{lemma}[Properties of Merge Operation]\label{lemma:mergeLH}
Let $r \in [2^k]$ be the output of a LH mechanism with budget $\epsilon$ using hash function $H_k$ (i.e., $r\overset{d}{=}LH(x,2^k,H_k,\epsilon)$). The merged result $r' = \text{Merge}(r)$ is distributionally equivalent to an LH output over domain $[2^{k-1}]$ with budget $\epsilon' = \ln((e^\epsilon + 1)/2)$. That is, 
$$\text{Merge}(r)\overset{d}{=}LH(x, 2^{k-1}, H_{k-1}, \ln((e^\epsilon + 1)/2)).$$
\end{lemma}
\begin{proof}
In the original domain $[2^k]$, let $p = \frac{e^\epsilon}{e^\epsilon + 2^k - 1}$ and $q = \frac{1}{e^\epsilon + 2^k - 1}$. 
For $r' \in [2^{k-1}]$, the probability of hitting the correct hash value is:
$\Pr[r' = H_{k-1}(x)] = \Pr[r \in \{2H_{k-1}(x)-1, 2H_{k-1}(x)\}] = p + q = \frac{e^\epsilon + 1}{e^\epsilon + 2^k - 1}$.
The probability of hitting any incorrect value $z \neq H_{k-1}(x)$ is $2q = \frac{2}{e^\epsilon + 2^k - 1}$.
The new privacy budget $\epsilon'$ satisfies $e^{\epsilon'} = \frac{p+q}{2q} = \frac{e^\epsilon + 1}{2}$. Thus, $\epsilon' = \ln(\frac{e^\epsilon + 1}{2})$.
\end{proof}

\begin{lemma}[Properties of Rescale Operation]\label{lemma:rescaleLH}
Let $r \in [2^k]$ be the output of a LH mechanism with budget $\epsilon$ using hash function $H_k$ (i.e., $r\overset{d}{=}LH(x,2^k,H_k,\epsilon)$). Then $\text{Rescale}(r, \beta)$ satisfies $\epsilon'$-LDP with:
\[
\text{Rescale}(r,\beta) \overset{d}{=} LH\left(x, 2^k, H_k, \ln\left(\frac{2^k\beta e^\epsilon + (1-\beta)(e^\epsilon + 2^k - 1)}{2^k\beta + (1-\beta)(e^\epsilon + 2^k - 1)}\right)\right).
\]
\end{lemma}

\textbf{Multi-ter Implementation.} Build upon previous results, we present the the Multi-tier Local Hash mechanism in Algorithm \ref{alg:multitierLH}, which generates middle-regime templates at reference levels where $e^{\epsilon} = 2^k - 1$ for $k\in [2, K]$. 

\begin{algorithm}[ht]
\caption{Multi-tier Local Hash Mechanism $LH^L(x)$}
\label{alg:multitierLH}
\KwIn{Privacy list $L=\{\epsilon_1,\dots,\epsilon_m\}$; true value $x \in [d]$}
\KwOut{Query results $\{r_i\}_{i \in [m]}$}

\tcp{\color{gray}1. Handle highest budget}

$K \leftarrow \lfloor \log_2(e^{\epsilon_1} + 1) \rfloor$\;

$s \leftarrow \text{LH}(x, 2^K, H_K, \epsilon_1)$ 

$\rho \leftarrow e^{\epsilon_1}$\;

$R \leftarrow \{(\epsilon_1, s, K)\}$\;

$k\leftarrow K$

\tcp{\color{gray}2. Generate middle-budget templates}

\While{$k > 1$}{
    $\rho' \leftarrow 2^k - 1$\;
    \If{$\rho > \rho'$}{
        $s \leftarrow \rescale(s, \frac{(\rho'-1)(\rho + 2^k - 1)}{d(\rho - \rho')})$\; 
        
        $\rho \leftarrow \rho'$ 

        append $(\ln \rho, s, k)$ to $R$\;
    }
    $s \leftarrow \Merge(s)$; \ \ $\rho \leftarrow (\rho + 1)/2$;

    $k \leftarrow k - 1$\;
    
}

\tcp{\color{gray}3. Handle low budget regime (k=1)}


    
append $(\ln \rho, s, k)$ to $R$\;

\For{$\epsilon_i\in L$ where $e^{\epsilon_i}< \rho$}{
        $\rho' \leftarrow e^{\epsilon_i}$

        $s \leftarrow \rescale(s, \frac{(\rho'-1)(\rho + 2^k - 1)}{d(\rho - \rho')})$\; 
        
        $\rho \leftarrow \rho'$ 

        append $(\ln \rho, s, k)$ to $R$\;
        
    }

\tcp{\color{gray}4. Template matching}

$\{r_1,\dots,r_m\} \leftarrow \text{TemplateMatch}(R, L)$\;
\Return{$\{r_1,\dots,r_m\}$}
\end{algorithm}

\textbf{Theoretical Analysis.}
For each middle-regime template in the hierarchy $k$ ($k\in [2,K]$), the hash domain size $2^k$ is optimal for its privacy level $\ln(2^k-1)$. The multiplicative error factor is upper bounded by the maximum ratio of variances between two consecutive templates (if we ignore the $\epsilon_1$).
}


\section{Discussions}\label{app:dis}

\textbf{Online privacy level requests. } In some settings, the analysts may not arrive simultaneously, and query systems might have to handle online queries with diverse privacy levels. There is a simple strategy for solving this issue by pre-computing the results using a fine-grained privacy list $L$. We generate a privacy list $L$ in advance that covers the range of upcoming requests and has fine-grained gaps (i.e., $\epsilon_i-\epsilon_{i+1}$ is sufficiently small). Then, we run our multi-tier algorithm on $L$ to obtain $\{r_i\}_{i\in L}$. For an incoming request with budget $\epsilon$, we find a result with a close but lower budget from $\{r_i\}_{i\in L}$. This strategy works for both multi-tier noise-adding mechanisms and the multi-tier Subset mechanism.

\textbf{Further template improvements. } When the privacy list $L$ is given and fixed, there is still some room for utility improvement in the template-based multi-tier approach. For example, in the multi-tier Subset mechanism, when the largest budget $\epsilon_1$ satisfies $\log((d+1)/2-1)<\epsilon_1< \log((d-1)/1-1)$, the $\textit{multi-tier Subset}^L_2$ cannot exactly match the ideal single-tier Subset mechanism, since there is still a 'hole' between $k=1$ and $k=2$. A simple strategy would be to skip the template $k=1$ with $\epsilon'_1=\log(d/1-1)$ and directly generate the result for $\epsilon$ to match the ideal one exactly. More generally, when the optimal subset sizes of the target privacy list are sparse in $[1,\lfloor d/2\rfloor]$, it is possible to skip some template $k$ without using the rescale operator, so as to maintain a higher budget (i.e., $\rho$ in Algorithm \ref{alg:multitiersubset}) for the next target privacy level. In some cases, this could improve the overall utility. However, this heuristic could also harm the utility for the next target level in other cases: in the previous example, if the next target level is $\epsilon_2=\log(d/2-1)$, transforming from the result for $\epsilon_1$ will not be able to achieve an exact optimal match for $\epsilon_2$ (if $\epsilon_1< \log((d-2)/1-1)$), whereas $\textit{multi-tier Subset}^L_2$ is capable of doing so. We leave such improvements for future research.

\textbf{Supporting more queries. } This work mainly studies the most fundamental multi-tier count queries in the curator model and the local model of DP. Considering diverse types of queries in multi-analyst databases, mobile services, and data markets, extending to other query types (e.g., numerical queries \cite{geng2015staircase}, selection queries \cite{mcsherry2007mechanism}) is of great interest, though challenging: (i) the staircase mechanism \cite{geng2015staircase} for numerical summation queries has optimal error rates but exhibits a staircase-style probability design with sharp transitions between stairs; (ii) numerical queries in LDP under the low privacy regime remain largely unexplored both from a theoretical perspective (e.g., optimal error rates) and in terms of optimal base mechanisms; (iii) selection queries with the exponential mechanism are considerably harder in the multi-tier setting, as the output space grows exponentially with the number of tiers. We leave these extensions for future work.